\documentclass[letterpaper,twocolumn,10pt]{article}
\usepackage{usenix}

\usepackage{times}
\usepackage{amsmath,amssymb}
\usepackage{amsthm}
\usepackage{booktabs}
\usepackage{algorithm}
\usepackage{algpseudocode}
\usepackage{graphicx}
\usepackage{xcolor}
\usepackage{tikz}
\usepackage{pgfplots}
\usepackage{url}
\usepackage{hyperref}
\usepackage{paralist}
\usepackage{balance}
\usetikzlibrary{arrows.meta,positioning,fit,backgrounds,shapes.geometric,calc}
\usepgfplotslibrary{groupplots}
\pgfplotsset{compat=1.18}

\definecolor{plotnavy}{HTML}{17365D}
\definecolor{plotblue}{HTML}{3B6FB6}
\definecolor{plotteal}{HTML}{148F8B}
\definecolor{plotamber}{HTML}{D18B22}
\definecolor{plotcoral}{HTML}{D95F59}
\definecolor{plotgray}{HTML}{687481}
\definecolor{plotgrid}{HTML}{E8EEF5}

\newcommand{\scheme}{\textsc{DHMark}}
\newcommand{\payload}{\mathsf{m}}
\newcommand{\ctx}{\mathsf{ctx}}
\newcommand{\pk}{\mathsf{pk}}
\newcommand{\sk}{\mathsf{sk}}

\newcommand{\Sign}{\mathsf{Sign}}
\newcommand{\Verify}{\mathsf{Verify}}
\newcommand{\Reg}{\mathsf{Reg}}

\newcommand{\Setup}{\mathsf{Setup}}

\newcommand{\Label}{\mathsf{Label}}
\newcommand{\PKCheck}{\mathsf{Check}}
\newcommand{\PKSF}{\mathsf{PKSF}}
\newcommand{\Adv}{\mathcal{A}}

\newtheorem{lemma}{Lemma}
\newtheorem{theorem}{Theorem}

\begin{document}

\date{}

\title{\Large \bf \scheme: Public-Key Watermarking for LLM-Generated Text \\ via Diffie--Hellman-Guided Rejection Sampling}

\author{
	{\rm Haocheng Fu\textsuperscript{1,2} \quad Yuqi Qian\textsuperscript{1,2} \quad Luyao Wang\textsuperscript{1,2} \quad Yun Cao\textsuperscript{1,2}} \\ \\
	\textsuperscript{1} Institute of Information Engineering, Chinese Academy of Sciences, Beijing, China \\
	\textsuperscript{2} School of Cyber Security, University of Chinese Academy of Sciences, Beijing, China
} 

\maketitle

\begin{abstract}
	Large language model (LLM) watermarking provides an important mechanism
	for tracing the provenance of generated text. Existing statistical
	watermarks are often effective and robust, but most of them rely on private
	detection keys, which centralizes verification and complicates public
	auditing. Recent public or publicly verifiable watermarking schemes improve
	key management, yet many of them rely on exact recovery of embedded
	cryptographic strings, making them fragile under token edits, truncation,
	copy-paste, and low-entropy generation. This paper introduces \scheme{}, a
	public-key watermarking framework for LLM-generated text. The key idea is to
	separate payload authorization from noisy textual evidence. An issuer signs
	a short registry payload bound to a public context, and the payload is
	expanded into many one-bit equations. During generation, a
	Diffie--Hellman-guided token-labeling interface assigns each candidate token
	a public equation vote, and the sampler softly or selectively promotes
	candidates whose votes agree with the authorized payload. During
	verification, third-party verifiers use public information to extract token
	votes, aggregate them into equation-level evidence, and score only signed
	registry records. This design avoids exact recovery of a long embedded
	signature and instead treats watermark detection as registry-aided
	statistical evidence aggregation. We formalize the public-verification
	setting, analyze label pseudorandomness, registry-backed soundness, and
	sampling distortion, and evaluate a prototype under truncation,
	substitution, copy-paste, wrong-context, and plain-generation attacks. In
	the default 32-bit configuration, \scheme{} maintains at least a $0.967$
	valid rate across eight edit conditions while yielding a $0.000$ acceptance
	rate on three negative controls.
\end{abstract}
	
\section{Introduction}

Large language models (LLMs) are rapidly becoming content infrastructure. Text generated by a model inside a closed platform can later be copied into social media posts, student submissions, legal drafts, news articles, code reviews, or even new training corpora. Once such text leaves the platform, its provenance is difficult to establish. A downstream verifier may need to determine whether a document contains text generated by a particular model service, but may not be willing or able to query the provider's private detection API, reveal the document being checked, or rely on a centralized decision. In many realistic settings, including content moderation, academic integrity, journalism, auditing, and dataset curation, provenance verification should therefore be possible offline, by third parties, using public information.

Watermarking is a natural mechanism for this purpose because it modifies the generation process itself rather than training an external classifier after the fact. In LLM text watermarking, the generator slightly changes the next-token sampling procedure so that the resulting token sequence carries a detectable signal. Unlike image or audio watermarking, however, the carrier is not a fixed high-dimensional signal with relatively stable redundancy. It is an autoregressive distribution whose entropy varies sharply across positions. Some decoding steps provide many plausible candidate tokens, while others are nearly deterministic. A practical LLM watermark must therefore balance several coupled requirements, including detectability, text quality, robustness to edits, payload capacity, and key management~\cite{liu2024survey}.

A first major line of work studies private-key statistical watermarks. The green-list watermark of Kirchenbauer et al.~\cite{kirchenbauer2023watermark} partitions the vocabulary into preferred and non-preferred token sets using a secret key and softly promotes preferred tokens during decoding. Detection is formulated as a statistical hypothesis test: the detector counts how often observed tokens fall into the preferred set and rejects the null hypothesis of ordinary generation when the bias is sufficiently strong. Subsequent reliability studies analyze robustness under paraphrasing, human rewriting, document mixing, span detection, and low-entropy prompts~\cite{kirchenbauer2024reliability}. This line of work established an important principle: a watermark detector does not necessarily need to recover a hidden message; it only needs enough token-level evidence to distinguish watermarked generation from ordinary generation.

A second line of work improves the quality, robustness, and statistical optimality of keyed watermarking. Distortion-free and low-distortion methods modify the sampler more carefully, for example through inverse-transform sampling, exponential-min sampling, Gumbel-style sampling, unbiased reweighting, or statistically optimized decision rules~\cite{kuditipudi2024robust,hu2024unbiased,christ2024undetectable,li2024statistical}. Other schemes further explore the tradeoff among detectability, robustness, and generation quality~\cite{giboulot2024watermax,watme2024}, while semantic or context-aware labeling targets robustness to paraphrasing and semantic drift~\cite{liu2024sir,ren2024semamark,guo2024context}. At the system level, SynthID-Text demonstrates that sampling-time watermarking can be deployed at very large scale with low latency overhead and without modifying model training~\cite{dathathri2024synthid}. These advances show that LLM watermarks can be practical and statistically powerful. Nevertheless, most such systems remain fundamentally private-key systems: the ability to detect the watermark depends on secret material that is closely tied to the ability to generate, simulate, or remove the watermark signal.

This private-key setting creates a key-management dilemma for public provenance. If the detection key remains secret, verification is centralized: third parties must trust the model provider or query an online detector. If the detection key is released, the same information that enables public detection may also help an adversary imitate the watermark or adapt generation to evade it. This tradeoff is acceptable for internal platform moderation, but it is less suitable for public accountability. A university, journalist, marketplace, court, or independent auditor should be able to verify watermark evidence without learning a secret embedding key and without sending the text back to the issuer.

Several recent works therefore move toward public or publicly verifiable LLM watermarks. UPV separates generation and detection through different neural networks, aiming to make detection public while preserving unforgeability~\cite{liu2024upv}. Publicly-detectable watermarking takes a more explicitly cryptographic route: it embeds a publicly verifiable signature into language-model output using rejection sampling and error correction, so that anyone can run the detection algorithm without secret information~\cite{fairoze2023publicly}. PVMark instead retains an underlying private-key detector but uses zero-knowledge proofs to make the detector's execution publicly auditable without revealing the secret key~\cite{duan2026pvmark}. These works identify the right high-level goal: watermark verification should be separable from watermark embedding authority.

However, a gap remains between cryptographic authenticity and statistical robustness. Signature-bearing public watermarks provide a clean authenticity story, but they often require reliable recovery of a long embedded string. This recovery objective is fragile under token substitutions, insertions, truncation, copy-paste, local desynchronization, and low-entropy decoding positions. In contrast, statistical watermarks naturally tolerate noisy evidence because detection aggregates many weak token-level signals, but their public-key story is weaker when detection depends on the same secret material used to generate the bias. Multi-bit and semantic watermarks, including position-allocated messages, multi-layer voting, payload-bearing, and topic- or signal-based methods, further explore capacity and robustness~\cite{yoo2024multibit,feng2025bimark,nemecek2026topic}, yet they do not fully resolve the tension among public verification, authorized payload binding, and noisy evidence aggregation.

A useful perspective comes from sampler-based steganography and distribution-preserving generation. Recent generative steganographic constructions argue that the sampling procedure is the appropriate place to reason about both security and distributional distortion: if sampling is modified carefully, one can encode information while controlling the output distribution~\cite{wang2025sparsamp}. Steganography and watermarking have opposite detectability goals. A steganographic object should look like an ordinary model sample to everyone except the receiver, whereas a watermarked object should later provide verifiable evidence to a detector. Nevertheless, both settings share the same technical pressure point: the method must be implemented as a sampler, must respect the model's next-token distribution, and should admit a security argument rather than only an empirical detector.

This paper asks the following question:

\begin{quote}
	\emph{Can we design an LLM watermark whose embedding authority is secret, whose verification is public, and whose robustness comes from noisy statistical evidence rather than exact recovery of a long embedded string?}
\end{quote}

We answer this question with \scheme{}, a public-key watermarking construction for LLM-generated text. \scheme{} is built around a public-key token-labeling interface guided by Diffie--Hellman-style cryptographic structure. During generation, the watermarking sampler assigns each candidate token an equation vote and biases the next-token distribution toward candidates whose votes agree with an authorized payload. During verification, third parties use public information to extract token-level votes and validate their statistical agreement with signed registry records. This replaces the shared keyed hash used in many private-key watermarks with a public-key-compatible voting interface, making explicit the separation between payload authorization and public verification.

The second key design choice is to change the role of the embedded payload. A direct public-key approach would embed a full digital signature or long authentication string into the generated text and require the verifier to recover it exactly. While cryptographically clean, this creates a long fragile bitstream. \scheme{} instead uses a short signed registry record. The issuer signs a compact payload bound to the public context, and the payload is expanded into many one-bit fountain-style equations. Each generated token contributes a noisy vote for one equation. Verification aggregates these votes and scores only signed registry records against the observed evidence. Thus the verifier does not ask whether the text decodes a hidden message without error; it asks whether the text contains sufficient independent evidence for an issuer-authorized payload.

This formulation deliberately treats LLM generation as a noisy channel. At high-entropy positions, the sampler can introduce stronger watermark bias because many plausible candidates are available. At low-entropy positions, forcing compatibility may harm text quality, so the sampler can use a softer promotion rule or fall back to ordinary generation. Edits, substitutions, truncation, and local desynchronization further remove or corrupt token votes. \scheme{} absorbs these effects through redundant one-bit equations, local transcript labeling, majority voting, and registry-aided scoring. Robustness therefore comes from repeated local evidence rather than from exact recovery of a contiguous payload.

Public verification also requires a careful statement of what is being verified. \scheme{} is not a document-integrity proof. If an adversary copies a sufficiently long watermarked span into another document, the copied span may still provide valid evidence. The correct claim is that the text contains issuer-authorized watermark evidence under a specified public context, not that every surrounding token is authentic or untampered. This distinction is central to our threat model and evaluation: copy-paste and malicious-suffix attacks are treated as provenance ambiguity rather than ordinary false positives on plain generation.

We implement \scheme{} on open-weight LLMs and evaluate it under token-preserving edits, truncation, substitution, copy-paste, wrong-context detection, and plain-generation false-positive tests. Our evaluation follows the joint quality--detectability--robustness perspective advocated by recent watermark benchmarks~\cite{piet2024markmywords,tu2024waterbench,liang2025waterpark}, but is organized around the proposed public-verification interface rather than only AI-text classification accuracy. In addition to detection validity, we report accepted equations, equation rank, unique equation coverage, registry agreement, fallback rate, candidate entropy, and sampling overhead. These measurements reveal when detection succeeds because enough independent evidence remains, and when it fails because generation is too low-entropy, too repetitive, or too heavily edited.

We introduce \scheme{}, a public-key LLM watermarking construction based on Diffie--Hellman-guided token labeling, which separates secret embedding authority from public watermark verification. Our contributions are as follows.
\begin{compactitem}
	\item \textbf{Registry-aided public-verification formulation.} We formalize public watermark verification as a test for noisy evidence supporting a \emph{signed, context-bound registry payload}. This separates payload-authorization authority from the public ability to assess textual evidence, and replaces brittle exact recovery of a long signature or message string with a bounded set of issuer-authorized hypotheses.
	\item \textbf{DH-guided redundant-evidence construction.} We construct a sampler that maps locally contextualized token candidates to public-key-compatible equation votes, expands a short payload into one-bit fountain-style equations, and aggregates repeated votes by majority and rank-aware registry scoring. The design operates on deployed top-$k$/top-$p$ distributions and uses soft promotion and fallback sampling when compatible mass is inadequate, so robustness comes from redundant local evidence rather than a contiguous encoded bitstream.
	\item \textbf{Explicit conditional security analysis.} We state the exact guarantee targeted by unpredictable labels, concentration of noisy votes, bounded sampling distortion, and registry-backed soundness for unsigned or unwatermarked text, and distinguish it from full document integrity or resistance to a detector-aware optimizer. We give direct proofs for the three supporting lemmas and defer the registry-backed soundness proof to the appendix.
	\item \textbf{Public-verification evaluation under transformations.} On 30 prompts per configuration, the default 32-bit system accepts all identity texts, retains valid rates from $0.967$ to $1.000$ across eight editing and reuse attacks, and rejects all evaluated plain and wrong-context controls. Our diagnostic measurements further expose the evidence coverage, fallback behavior, and sampling cost behind each decision.
\end{compactitem}

\section{Preliminaries and Problem Formulation}
\label{sec:preliminaries}

This section introduces the notation and design requirements used throughout the paper. We first formalize autoregressive LLM decoding and token-level watermarking. We then discuss the distinction between private detection, public detection, and public verifiability. Finally, we introduce the cryptographic and coding primitives needed by \scheme{} and state the problem addressed in this work.

\subsection{Autoregressive LLM Decoding}
\label{subsec:llm_decoding}

Let $M$ be an autoregressive language model with vocabulary $V$. Given a prompt $p$ and a generated prefix $x_{<t}=(x_1,\ldots,x_{t-1})$, the model defines a next-token distribution
$$P_t(x)=\Pr_M[X_t=x\mid p,x_{<t}], \qquad x\in V.$$

In principle, a token can be sampled directly from $P_t$. In practical LLM deployments, however, the raw distribution is usually modified by temperature scaling and then truncated by decoding filters such as top-$k$, nucleus sampling, or their variants. We denote the resulting candidate set at position $t$ by
$$C_t \subseteq V.$$

The deployed sampling distribution is the normalized distribution over this candidate set:
$$
\widetilde P_t(x)=
\frac{P_t(x)\mathbf{1}[x\in C_t]}
{\sum_{x'\in C_t}P_t(x')}.
$$

This distinction is important for watermarking. A watermarking algorithm that forces low-probability tokens from the full vocabulary may produce unnatural text, especially at low-entropy positions. Therefore, in this paper we treat $\widetilde P_t$, rather than the untruncated $P_t$, as the operational distribution that the watermarking sampler should respect.

The entropy and probability mass of $C_t$ can vary significantly across positions. In open-ended or semantically flexible contexts, $C_t$ may contain many plausible tokens with relatively balanced probabilities. In deterministic contexts, however, a few tokens may dominate most of the probability mass. This variation creates a central challenge for LLM watermarking: the watermark should be strong enough to support reliable detection, but not so strong that it overrides the model's natural distribution.

\subsection{Token-Level Text Watermarking}
\label{subsec:token_watermarking}

A token-level LLM watermark modifies the generation process so that the resulting token sequence contains a detectable signal. A generic watermarking scheme consists of a generation algorithm and a detection algorithm. The generation algorithm produces a text sequence
$$
\mathsf{WGen}(M,p,\cdot)\rightarrow X,
$$
where $X=(x_1,\ldots,x_T)$. The detection algorithm outputs a binary decision
$$
\mathsf{Detect}(X,\cdot)\rightarrow \{0,1\}.
$$

A common design is statistical watermarking. At each generation step, the vocabulary or candidate set is partitioned into preferred and non-preferred tokens using a keyed pseudorandom function. The sampler then slightly increases the probability of preferred tokens. Detection counts whether the observed tokens fall into the preferred set more often than expected under ordinary generation. The green-list watermark of Kirchenbauer et al.~\cite{kirchenbauer2023watermark} is a representative example of this paradigm, and subsequent studies analyze its robustness under paraphrasing, document mixing, span detection, and low-entropy prompts~\cite{kirchenbauer2024reliability}.

Another line of work aims to reduce distributional distortion. Distortion-free, unbiased, or low-distortion watermarks modify the sampler more carefully so that the generated distribution remains close to the original model distribution while still producing a detectable keyed signal~\cite{kuditipudi2024robust,hu2024unbiased,christ2024undetectable,giboulot2024watermax}. These methods are especially relevant to practical deployment because watermarks should not noticeably degrade generation quality.

A third class of schemes embeds multiple bits or structured payloads into generated text. Such payloads may encode user identifiers, metadata, or cryptographic material. However, exact message recovery is difficult in the text domain. Token substitutions, insertions, deletions, paraphrasing, and truncation can easily corrupt a long bitstream. This motivates a different view adopted in \textsc{DHMARK}: the text is treated as a noisy source of token-level evidence, and verification is performed by aggregating many weak votes rather than by exactly recovering a long hidden message.

\subsection{Private Detection, Public Detection, and Public Verifiability}
\label{subsec:public_verification}

The key-management model is a fundamental part of an LLM watermarking scheme. In a private-key watermark, the detector uses a secret key to test whether a text is watermarked. This setting is simple and powerful, but it centralizes verification. A third party must either obtain the secret key or query the issuer's private detector. The former may enable forgery or adaptive removal, while the latter requires trust in an online service and may expose the text being checked.

Public detection aims to let anyone run the detection algorithm using public information. However, public detection alone is not sufficient. If the public detection rule also enables an adversary to construct watermarked text, then detection no longer proves that the text was generated by an authorized issuer. Thus, a stronger goal is public verifiability: third parties should be able to verify watermark evidence, while the authority to create valid watermark evidence should remain controlled.

Recent public or publicly verifiable watermarking methods address this problem in different ways. UPV separates generation and detection through different neural networks and targets unforgeable public verification~\cite{liu2024upv}. Publicly-detectable watermarking embeds publicly verifiable cryptographic information into language-model outputs~\cite{fairoze2023publicly}. PVMark uses zero-knowledge proofs to make the execution of a private-key detector publicly auditable without revealing the detection key~\cite{duan2026pvmark}. More generally, verifiable random functions show how a secret evaluation capability can be paired with publicly checkable outputs~\cite{micali1999vrf}. These works show that public verification is a meaningful and necessary goal, but they also expose a tension between cryptographic authenticity and robustness to noisy text transformations.

In this work, we adopt a registry-backed public-verification model. The issuer authorizes a short payload by signing a registry record. The generated text does not need to contain an exactly recoverable signature. Instead, it should contain enough token-level evidence that is statistically consistent with one signed registry payload. Therefore, the claim made by detection is deliberately limited: \emph{the text contains issuer-authorized watermark evidence}, rather than \emph{the entire document is an untampered issuer-generated object}. 

This distinction is essential for copy-paste and malicious-suffix scenarios, where a valid watermarked span may appear inside a larger untrusted document.

\subsection{Cryptographic Ingredients}
\label{subsec:crypto_ingredients}

\scheme{} uses two cryptographic ingredients: digital signatures and a public-key token-labeling interface.

First, digital signatures are used to authenticate registry records. Let $(\mathsf{sk}_{\mathsf{sig}},\mathsf{pk}_{\mathsf{sig}})$ be a signing key pair. For a payload $m$ and public context $ctx$, the issuer creates a signed record
$$
R=(m,\ell,H(ctx),\mathsf{fp}(\mathsf{pk}_{\mathsf{sig}}),
\mathsf{fp}(\mathsf{pk}_{\mathsf{samp}}),\sigma).
$$
The signature is computed as
$$
\sigma \leftarrow
\mathsf{Sign}_{\mathsf{sk}_{\mathsf{sig}}}
(m,\ell,H(ctx),\mathsf{fp}(\mathsf{pk}_{\mathsf{sig}}),
\mathsf{fp}(\mathsf{pk}_{\mathsf{samp}})).
$$
The signature prevents an adversary from introducing arbitrary payload claims after observing a text.

Second, the sampler assigns each token candidate a cryptographic vote. For a token candidate $x$ at position $t$, we define a local transcript
$$
c_t=(x_{t-h},\ldots,x_{t-1}),
$$
where $h$ is a small context window. The token item to be labeled is
$$
\tau_t(x)=(ctx,c_t,x).
$$
A token-labeling function maps this item to an equation index and a bit:
$$
(i_x,b_x)=\mathsf{Label}(ctx,c_t,x), \qquad i_x\in [N],\quad b_x\in\{0,1\}.
$$

In \scheme{}, this label is guided by public-key cryptographic structure. At an abstract level, the generation algorithm uses sampling authority to bias token selection toward labels compatible with an authorized payload, while the verifier uses public information to extract or validate token votes.

This public-key labeling interface must be interpreted carefully. If token labels are fully public and efficiently computable by anyone, then public verification is simple, but an adversary may also use the same labels to simulate watermark-compatible generation for an already authorized payload. In that case, unforgeability should be understood as registry-backed soundness: the adversary cannot create evidence for an unsigned payload or wrong context, except through statistical false positives. Stronger claims that only the issuer can ever generate accepting texts require additional cryptographic mechanisms, such as publicly verifiable trapdoor relations, proof-carrying labels, zero-knowledge detector proofs, or system-level access control. The present paper focuses on the registry-backed public-verification setting and explicitly states this security boundary.

\subsection{One-Bit Equations and Noisy Evidence Aggregation}
\label{subsec:fountain_equations}

To avoid exact recovery of a long embedded message, \scheme{} expands a short payload into many one-bit equations. Let $m=(m_1,\ldots,m_\ell)\in\{0,1\}^{\ell}$ be the payload authorized by the issuer. A public binary matrix
$$
G\in\{0,1\}^{N\times \ell}
$$
is derived from the public context $ctx$. The equation vector is
$$
Y(m)=Gm \bmod 2,
$$
where
$$
Y(m)=(y_1(m),\ldots,y_N(m)).
$$
Equivalently, each equation index $i\in[N]$ corresponds to a support set
$$
S_i\subseteq[\ell],
$$
and the target equation bit is
$$
y_i(m)=\bigoplus_{j\in S_i}m_j.
$$
During generation, each emitted token contributes one vote
$$
x_t \mapsto (i_t,b_t).
$$
Ideally, a watermarked token satisfies
$$
b_t=y_{i_t}(m).
$$
In practice, the vote may be noisy because of soft sampling, fallback generation, token substitutions, deletions, or local desynchronization. The verifier therefore aggregates votes by equation index. For each $i\in[N]$, define the vote multiset
$$
V_i(X)=\{b_t: i_t=i,\ 1\leq t\leq T\}.
$$
If $|V_i(X)|$ is large enough, the verifier estimates the equation bit by majority voting:
$$
\widehat y_i(X)=\mathsf{Majority}(V_i(X)).
$$
Let
$$
A(X)=\{i: |V_i(X)|\geq v_{\min}\}
$$
be the set of retained equations. For a signed registry payload $m_R$, the agreement score is
$$
\mathsf{score}(X,m_R)=
\frac{1}{|A(X)|}
\sum_{i\in A(X)}
\mathbf{1}[\widehat y_i(X)=y_i(m_R)].
$$

An ordinary unwatermarked text should agree with any fixed payload at a rate close to $1/2$, while a watermarked text should have a score noticeably above $1/2$. Detection is therefore formulated as noisy evidence aggregation against signed registry records.

\subsection{Problem Statement and Design Goals}
\label{subsec:problem_statement}

We consider a public-verification setting with four roles: an issuer, a generator, a registry, and a verifier. The issuer authorizes payloads by signing registry records. The generator runs the LLM and applies the watermarking sampler. The registry stores signed payload records. The verifier receives a candidate text, public keys, the public context, and registry records, and then decides whether the text contains sufficient evidence for an authorized payload.

A public watermarking scheme in this setting consists of four algorithms:
$$
\Pi=(\mathsf{KeyGen},\mathsf{Authorize},\mathsf{WGen},\mathsf{Detect}).
$$
The key-generation algorithm outputs signing and sampling keys:
$$
\mathsf{KeyGen}(1^\lambda)\rightarrow
(\mathsf{sk}_{\mathsf{sig}},\mathsf{pk}_{\mathsf{sig}},
\mathsf{sk}_{\mathsf{samp}},\mathsf{pk}_{\mathsf{samp}}).
$$
The authorization algorithm creates a signed registry record:
$$
\mathsf{Authorize}(\mathsf{sk}_{\mathsf{sig}},m,ctx)\rightarrow R.
$$
The watermarked generation algorithm produces a token sequence:
$$
\mathsf{WGen}(M,p,ctx,R,\mathsf{sk}_{\mathsf{samp}})\rightarrow X.
$$
The detection algorithm uses only public information and registry records:
$$
\mathsf{Detect}(X,ctx,\mathsf{Reg};
\mathsf{pk}_{\mathsf{sig}},\mathsf{pk}_{\mathsf{samp}})
\rightarrow \{0,1\}.
$$
The design goals are as follows.

\textbf{Public verification.}
The detector should operate offline using public keys, public context, and signed registry records. It should not require access to the issuer's private detector or sampling secret.

\textbf{Registry-backed authenticity.}
A text should be accepted only if its token-level evidence is consistent with a valid signed registry payload under the correct public context. Detection should reject wrong-context texts, unsigned payload claims, and ordinary plain generations except with small false-positive probability.

\textbf{Robust noisy evidence.}
The verifier should not rely on exact recovery of a long hidden bitstream. Instead, it should aggregate many local token votes and tolerate noise caused by fallback sampling, truncation, substitution, deletion, and copy-paste.

\textbf{Quality-preserving sampling.}
The watermarking sampler should respect the deployed decoding distribution $\widetilde P_t$. It should preferably operate within the top-$k$ or top-$p$ candidate set and avoid forcing low-probability tokens that would degrade fluency or semantic quality.

\textbf{Bounded verification semantics.}
Detection should be interpreted as evidence that the text contains issuer-authorized watermark signal. It should not be interpreted as a proof of full-document integrity or as a guarantee that every token was generated by the issuer.

Based on these definitions and goals, the next section presents \scheme{}, a public-key watermarking construction that combines token-level public-key votes, soft or hybrid sampling-time promotion, one-bit equation expansion, and registry-aided evidence scoring.

\section{\scheme{} Method}
\label{sec:method}

This section gives the full construction of \scheme{}. We begin by formalizing
the public-verification setting, since it is the point at which our goal
departs from both private-key statistical watermarking and classical
steganographic communication. We then present the public-key sampling function,
the registry-backed payload layer, the one-bit equation encoding, the embedding
and verification algorithms, and the security properties that follow from these
components.

\paragraph{Roadmap.}
Section~\ref{sec:method:setting} defines the public verification experiment and
the notation used throughout the construction. Section~\ref{sec:method:overview}
gives the construction overview. Section~\ref{sec:method:pksf} introduces the
public-key sampling function and its Diffie--Hellman instantiation.
Section~\ref{sec:method:eqs} expands each payload into redundant
one-bit equations for multi-position voting. Section~\ref{sec:method:registry}
explains how registry signatures authenticate the watermark evidence.
Section~\ref{sec:method:analysis} summarizes the resulting security claims,
design tradeoffs, and limits.

\subsection{Watermarking and Public Verification Setting}
\label{sec:method:setting}

\paragraph{Language-model channel.}
Let \(\mathcal{M}\) be an autoregressive language model with vocabulary
\(\mathcal{V}\). Given a prompt \(p\) and generated prefix
\(x_{<t}=(x_1,\ldots,x_{t-1})\), the model induces a next-token distribution
\[
    P_t(x)=\Pr_{\mathcal{M}}[X_t=x\mid p,x_{<t}],
    \qquad x\in\mathcal{V}.
\]
Modern LLM deployments rarely sample from \(P_t\) directly. They first apply a
decoding filter, such as top-\(k\), top-\(p\), temperature scaling, or a
combination of them. We model this filter by a candidate set
\(C_t\subseteq\mathcal{V}\) and write the normalized deployment distribution as
\[
    \widetilde{P}_t(x)
    =
    \frac{P_t(x)\mathbf{1}[x\in C_t]}
         {\sum_{x'\in C_t}P_t(x')}.
\]
The design requirement is that watermarking should modify this deployed
sampler, not force arbitrary low-probability tokens from the full vocabulary.
This requirement is essential in practice because the probability mass and
entropy of \(C_t\) vary sharply across positions.

\paragraph{Watermark as token-level evidence.}
In \scheme{}, a watermark is not a literal message appended to the text and not
a separate metadata field. It is a cryptographically structured bias introduced
during token sampling. The issuer authorizes a short payload
\(\payload\in\{0,1\}^{\ell}\). This payload determines a vector of target
watermark bits
\[
    Y(\payload)=(y_1(\payload),\ldots,y_N(\payload)),
    \qquad y_i(\payload)\in\{0,1\}.
\]
At each generation step, every candidate token \(x\in C_t\) is assigned a
publicly checkable watermark label
\[
    L(\tau_t(x))=(i_x,b_x),
    \qquad i_x\in[N],\ b_x\in\{0,1\}.
\]
The generator prefers tokens whose label agrees with the authorized payload:
\[
    b_x = y_{i_x}(\payload).
\]
Therefore each emitted token acts as a noisy vote for one equation bit of the
payload. A watermarked text is a sequence whose token labels agree with one
authorized payload noticeably more often than ordinary sampling would produce.
This is the central watermark signal.

This viewpoint also clarifies the difference between watermarking and exact
message extraction. The detector does not need to recover \(\payload\) exactly
from the text. Instead, it evaluates whether the collected token votes provide
statistically sufficient evidence for some signed payload in the registry. If
\(\widehat{y}_i(X)\) denotes the observed majority bit for equation \(i\), then
the core evidence statistic has the form
\[
    \mathsf{score}(X,\payload)
    =
    \frac{1}{|A(X)|}
    \sum_{i\in A(X)}
    \mathbf{1}[\widehat{y}_i(X)=y_i(\payload)],
\]
where \(A(X)\subseteq[N]\) is the set of equation indices for which the text
contains enough votes. Ordinary unwatermarked text should satisfy
\(\mathbb{E}[\mathsf{score}(X,\payload)]\approx 1/2\) for any fixed authorized
payload, while watermarked text should have score significantly above \(1/2\).
The watermark detection problem is thus a public hypothesis test over token
labels:
\[
\begin{aligned}
    H_0 &: X \text{ was sampled without payload-dependent bias},\\
    H_1 &: X \text{ contains token-level evidence for an authorized } \payload.
\end{aligned}
\]

\paragraph{Local public transcript.}
The public context is denoted by \(\ctx\). It binds the watermark instance to
the prompt, model identifier, sampling parameters, public-key fingerprint, and
application metadata. At token position \(t\), \scheme{} also uses a bounded
local transcript
\[
    c_t=(x_{t-h},\ldots,x_{t-1}),
\]
where \(h\) is a small context length and missing prefix elements are replaced
by a fixed domain-separation symbol. The token item labeled by the cryptographic
sampler is therefore
\[
    \tau_t(x)=(\ctx,c_t,x),
\]
rather than an absolute-position item \((\ctx,t,x)\). This choice is not merely
an implementation detail. If an adversary deletes or inserts a short span, an
absolute-position watermark can remain desynchronized for the rest of the
document. In contrast, a local-context watermark loses at most the edited span
and the following \(h\) positions before the transcript seen by the verifier
realigns with the transcript used by the generator.

\paragraph{Parties and interfaces.}
The public-verification setting contains four roles.
The \emph{issuer} authorizes payloads and holds a signing key
\(\sk_{\mathsf{sig}}\). The \emph{generator} runs the LLM and holds the secret
sampling key \(\sk_{\mathsf{samp}}\). The \emph{registry} stores signed payload
records. The \emph{verifier} receives a text \(X=(x_1,\ldots,x_T)\), public
parameters, and the registry, and decides whether \(X\) contains sufficient
watermark evidence for an authorized payload. The verifier is offline in the
sense that it does not need the generator's sampling secret or an online oracle
owned by the issuer; it only needs public keys and registry records.

Formally, a public watermarking scheme in this setting has algorithms
\[
    \Pi=(\mathsf{KeyGen},\mathsf{Authorize},
    \mathsf{WGen},\mathsf{Detect}).
\]
\(\mathsf{KeyGen}(1^\lambda)\) outputs signing keys
\((\sk_{\mathsf{sig}},\pk_{\mathsf{sig}})\) and sampling keys
\((\sk_{\mathsf{samp}},\pk_{\mathsf{samp}})\). \(\mathsf{Authorize}\) creates a
registry record \(R\) for a payload \(\payload\) and context \(\ctx\).
\(\mathsf{WGen}\) uses \(\mathcal{M}\), \(\sk_{\mathsf{samp}}\), \(R\), and
\(\ctx\) to generate watermarked text. \(\mathsf{Detect}\) uses only
\((\pk_{\mathsf{sig}},\pk_{\mathsf{samp}})\), \(\ctx\), the registry, and the
candidate text to output
\[
    \mathsf{Detect}(X,\ctx,\Reg;\pk_{\mathsf{sig}},\pk_{\mathsf{samp}})
    \in \{0,1\}.
\]
This separation is the central public-key property: the secret used to create
evidence is not needed to check evidence.

\begin{figure*}[t]
	\centering
	\includegraphics[width=\textwidth]{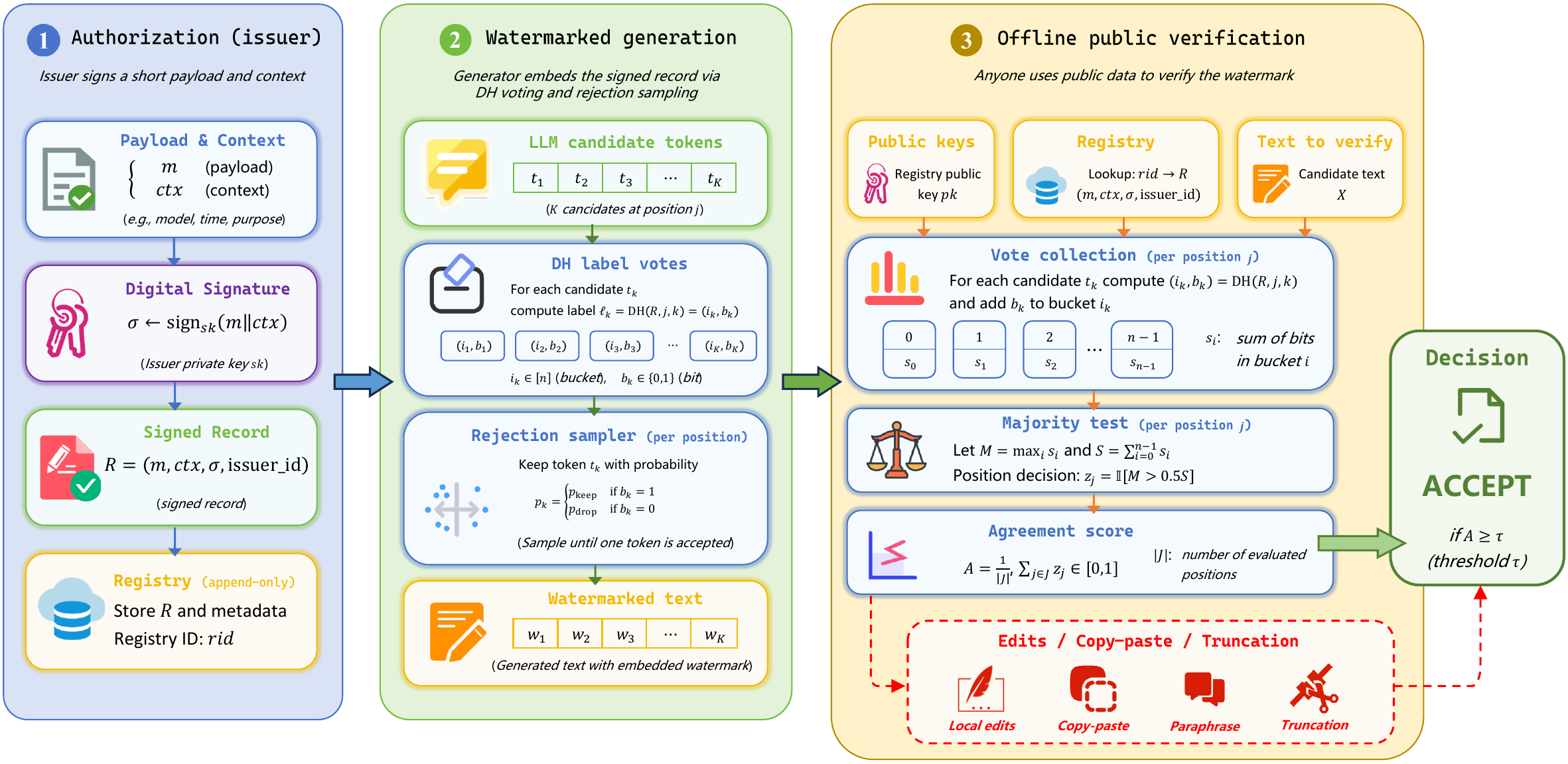}
	\caption{\scheme{} workflow. The generator uses secret sampling material to
		create token-level votes for a signed payload, while the verifier uses public
		keys and registry records to aggregate evidence and authenticate the watermark.}
	\label{fig:dhmark-overview}
\end{figure*}

\paragraph{Adversarial experiment.}
The adversary \(\Adv\) knows the scheme, the public keys, the public context,
and the detector. It may observe watermarked texts, generate ordinary
unwatermarked texts, and apply editing transformations
\[
    X' \leftarrow \mathcal{T}(X),
\]
including truncation, burst deletion, token substitution, copy-paste, and
suffix appending. It does not know \(\sk_{\mathsf{sig}}\) or
\(\sk_{\mathsf{samp}}\), and it cannot insert unsigned records into \(\Reg\).
The two core events are
\[
    \mathsf{FN}_{\mathcal{T}}
    =
    \Pr_{X\leftarrow \mathsf{WGen}(p,\ctx,R)}
    [\mathsf{Detect}(\mathcal{T}(X),\ctx,\Reg)=0],
\]
and
\[
    \mathsf{FP}
    =
    \Pr_{Y\leftarrow \mathsf{Gen}_{\mathcal{M}}(p)}
    [\mathsf{Detect}(Y,\ctx,\Reg)=1].
\]
Robustness asks that \(\mathsf{FN}_{\mathcal{T}}\) remain small for benign or
moderate edits \(\mathcal{T}\). Soundness asks that \(\mathsf{FP}\) be small on
ordinary unwatermarked texts and on texts generated under the wrong public
context.

\paragraph{What is being authenticated?}
The detector does not prove that every token in a document was generated by the
issuer. Instead, it proves that the text contains enough independent evidence
for an issuer-authorized payload. This distinction matters. Let \(A(X)\) be the
set of equation indices for which the verifier obtains reliable observations.
Detection is evidence aggregation:
\[
\begin{aligned}
    \exists R\in\Reg \quad \text{s.t.}\quad
    &\Verify_{\pk_{\mathsf{sig}}}(R,\ctx)=1,\\
    &\mathsf{Evidence}(X,R,\ctx)\geq \theta .
\end{aligned}
\]
Thus a valid detection means ``this text contains watermark evidence consistent
with a signed registry payload,'' not ``the full document is an untampered
signed object.'' Copy-paste attacks are therefore treated as provenance
ambiguity rather than cryptographic signature forgery.

\paragraph{Immediate consequences.}
The setting above yields three useful design constraints. First, public
verification requires a sampling label that can be checked from public data;
otherwise the verifier degenerates into a private-key detector. Second,
unforgeability should be tied to authorized payload records rather than to the
raw bit string decoded from text, because token-level extraction is noisy.
Third, robustness should come from repeated local evidence rather than from a
single long hidden message: if each token supplies a bounded, locally
self-synchronizing vote, then deletions and truncations reduce the amount of
evidence but need not corrupt the entire watermark.

\subsection{Overview}
\label{sec:method:overview}

\scheme{} has three layers, shown in Figure~\ref{fig:dhmark-overview}. The
authentication layer authorizes a short payload through a signed registry
record. The coding layer expands the payload into many one-bit equations. The
sampling layer uses a DH-guided label for each candidate token and performs
rejection sampling so that generated tokens preferentially vote for the
authorized equation bits. Verification runs the same pipeline in reverse:
extract public token votes, aggregate them into noisy equation estimates, and
score only signed registry payloads.

The generation side is shown in Algorithm~\ref{alg:embed}. It never embeds a
large contiguous packet. Instead, each token position receives a local
transcript \(c_t\), each candidate token receives a label \((i,b)\), and the
sampler conditions on candidates compatible with the authorized equation vector
\(Y(\payload)\). Low-entropy positions may have no compatible candidate inside
the deployed top-\(k\)/top-\(p\) set; in that case the sampler falls back to
ordinary sampling to preserve text quality.

The verification side is shown in Algorithm~\ref{alg:verify}. It reconstructs
public token votes, applies majority voting for each equation, and then scores
only signed registry records. This is why verification can succeed even when
exact payload decoding fails: the detector asks whether the observed evidence
is sufficiently aligned with an authorized payload, not whether every hidden
bit was recovered without error.

\begin{algorithm}[t]
	\caption{\scheme{} watermarked generation}
	\label{alg:embed}
	\small
	\begin{algorithmic}[1]
		\Require Prompt \(p\), model \(\mathcal{M}\), context \(\ctx\), registry \(\Reg\),
		sampling secret \(\sk_{\mathsf{samp}}\), signing key \(\sk_{\mathsf{sig}}\),
		length \(T\)
		\Ensure Watermarked token sequence \(X=(x_1,\ldots,x_T)\), registry record \(R\)
		\State Sample payload \(\payload\in\{0,1\}^{\ell}\)
		\State Create signed record \(R\leftarrow
		\Sign_{\sk_{\mathsf{sig}}}(\payload,\ell,H(\ctx))\) and publish it to \(\Reg\)
		\State Derive \(G\) from \(\ctx\) and compute \(Y(\payload)=G\payload\bmod 2\)
		\State \(X\leftarrow [\,]\)
		\For{\(t=1\) to \(T\)}
		\State Compute \(P_t=\mathcal{M}(\cdot\mid p,X)\) and candidate
		distribution \(\widetilde{P}_t\) over \(C_t\)
		\State Set local transcript \(c_t\leftarrow(x_{t-h},\ldots,x_{t-1})\)
		\State \(M_t\leftarrow \{x\in C_t:
		\Label_{\sk_{\mathsf{samp}}}(\ctx,c_t,x)=(i,b)
		\land b=y_i(\payload)\}\)
		\If{\(M_t\neq\emptyset\)}
		\State Sample \(x_t\sim \widetilde{P}_t(\cdot\mid M_t)\)
		\Else
		\State Sample \(x_t\sim \widetilde{P}_t\) \Comment{quality-preserving fallback}
		\EndIf
		\State Append \(x_t\) to \(X\)
		\EndFor
		\State \Return \(X,R\)
	\end{algorithmic}
\end{algorithm}

\subsection{Public-Key Sampling Function}
\label{sec:method:pksf}

The cryptographic primitive used by \scheme{} is a public-key sampling function
\[
    \PKSF = (\Setup,\Label,\PKCheck).
\]
\(\Setup(1^\lambda)\) outputs a secret key \(\sk_{\mathsf{samp}}\) and public
key \(\pk_{\mathsf{samp}}\). For a token transcript item
\(\tau=(\ctx,c,x)\), the secret labeling algorithm outputs an equation index
and bit
\[
    (i,b) \leftarrow \Label_{\sk_{\mathsf{samp}}}(\tau),
    \qquad i \in [N],\ b\in\{0,1\},
\]
where \(N\) is the number of public equations. The public checking algorithm
\[
    \PKCheck_{\pk_{\mathsf{samp}}}(\tau,i,b) \in \{0,1\}
\]
allows the verifier to recover or validate the token vote.

\begin{algorithm}[t]
	\caption{\scheme{} public verification}
	\label{alg:verify}
	\small
	\begin{algorithmic}[1]
		\Require Text \(X=(x_1,\ldots,x_T)\), context \(\ctx\), public keys
		\(\pk_{\mathsf{samp}},\pk_{\mathsf{sig}}\), registry \(\Reg\)
		\Ensure Accept or reject
		\State Derive \(G\) from \(\ctx\)
		\State Initialize vote multisets \(V_i\leftarrow \emptyset\) for all \(i\in[N]\)
		\For{\(t=1\) to \(T\)}
		\State Set local transcript \(c_t\leftarrow(x_{t-h},\ldots,x_{t-1})\)
		\For{\(i=1\) to \(N\), \(b\in\{0,1\}\)}
		\If{\(\PKCheck_{\pk_{\mathsf{samp}}}((\ctx,c_t,x_t),i,b)=1\)}
		\State \(V_i\leftarrow V_i\cup\{b\}\)
		\EndIf
		\EndFor
		\EndFor
		\State \(A\leftarrow\{i: |V_i|\geq v_{\min}\}\)
		\State \(\widehat{y}_i\leftarrow \mathsf{majority}(V_i)\) for all \(i\in A\)
		\ForAll{signed records \(R\in\Reg\)}
		\If{\(\Verify_{\pk_{\mathsf{sig}}}(R,\ctx)=0\)}
		\State \textbf{continue}
		\EndIf
		\State Compute \(Y(\payload_R)=G\payload_R\bmod 2\)
		\State Compute \(\mathsf{score}(X,\payload_R)\) over retained equations \(A\)
		\If{\(\mathsf{score}\geq\theta\) and evidence gates pass}
		\State \Return Accept with record \(R\)
		\EndIf
		\EndFor
		\State \Return Reject
	\end{algorithmic}
\end{algorithm}

\paragraph{Correctness requirement.}
For every transcript item \(\tau\), the public checker must agree with the
secret label:
\[
    \Pr\!\left[
        \PKCheck_{\pk_{\mathsf{samp}}}
        (\tau,\Label_{\sk_{\mathsf{samp}}}(\tau))=1
    \right]=1.
\]
For all other pairs \((i,b)\neq \Label_{\sk_{\mathsf{samp}}}(\tau)\), the
checker should reject except with negligible probability. This gives the
verifier the same token vote that the generator used, without revealing the
secret sampling key.

\paragraph{Sampling pseudorandomness.}
To preserve generation quality and prevent public keys from enabling trivial
watermark fabrication, labels should be computationally indistinguishable from
uniform:
\[
    \Label_{\sk_{\mathsf{samp}}}(\tau)
    \approx_c
    (I,B),
    \qquad I\leftarrow [N],\ B\leftarrow\{0,1\},
\]
for any party that does not hold \(\sk_{\mathsf{samp}}\). Under this condition,
before conditioning on the target payload, each candidate token appears to vote
for a random equation bit. The watermark signal is introduced only by the
generator's rejection-sampling preference.

\paragraph{Diffie--Hellman instantiation.}
The intended instantiation works over a cyclic group \(\mathbb{G}\) of prime
order \(q\) with generator \(g\). The sampling key is
\[
    \sk_{\mathsf{samp}}=a\leftarrow\mathbb{Z}_q,
    \qquad
    \pk_{\mathsf{samp}}=A=g^a.
\]
A hash-to-group function maps each transcript item \(\tau\) to
\[
    U_\tau=H_{\mathbb{G}}(\tau)\in\mathbb{G}.
\]
The secret label is derived from the DH value
\[
    z_\tau = U_\tau^a,
    \qquad
    (i,b)=\mathsf{Hash}_{\mathsf{lab}}(z_\tau,\tau).
\]
By the DDH assumption, the tuple
\[
    (g,A,U_\tau,U_\tau^a)
\]
is computationally indistinguishable from
\[
    (g,A,U_\tau,g^r),\qquad r\leftarrow\mathbb{Z}_q,
\]
for fresh \(\tau\). Hence the derived label is pseudorandom to any party that
cannot compute the DH value.

For public verification, \(\PKCheck\) must be realized by a DH-compatible public
checking mechanism. One abstract way to write the required interface is
\[
    \PKCheck_{\pk_{\mathsf{samp}}}(\tau,i,b)=1
    \iff
    (i,b)=\mathsf{Hash}_{\mathsf{lab}}(U_\tau^a,\tau),
\]
while allowing the verifier to test this relation using only
\(\pk_{\mathsf{samp}}\) and public transcript data. This can be instantiated by
a publicly verifiable trapdoor-PRF interface, a pairing-based relation, or a
compact proof that the vote was derived consistently from the DH label. The
watermarking construction below only assumes this \(\PKSF\) interface; its
statistical encoding and authentication layers are independent of the concrete
public-checking realization.

\subsection{Multi-Position Voting and Error \\ Correction}
\label{sec:method:eqs}

The payload is expanded into \(N\) public one-bit equations. For each equation
index \(i\in[N]\), public randomness derived from \((\ctx,i)\) selects a support
set
\[
    S_i\subseteq[\ell], \qquad |S_i|=d_i,
\]
where \(d_i\) is the equation degree. The target equation bit is
\[
    y_i(\payload)=\bigoplus_{j\in S_i} m_j.
\]
Equivalently, let \(G\in\{0,1\}^{N\times \ell}\) be a public binary matrix with
\(G_{i,j}=1\) iff \(j\in S_i\). Then
\[
    Y(\payload)=G\payload \pmod 2.
\]

During generation, token \(x_t\) contributes a vote \((i_t,b_t)\). Ideally,
\[
    b_t = y_{i_t}(\payload),
\]
but fallback sampling, edits, token substitutions, and local desynchronization
can flip or remove votes. Let
\[
    V_i(X)=\{b_t: i_t=i,\ 1\leq t\leq T\}
\]
be the multiset of votes observed for equation \(i\). The verifier forms
\[
    \widehat{y}_i(X)=
    \mathsf{majority}(V_i(X))
\]
when \(|V_i(X)|\geq v_{\min}\), and discards equation \(i\) otherwise. Thus
the redundancy is not a packet-level error-correcting code that requires exact
payload recovery. It is a many-position voting code: multiple token positions
estimate the same equation bit, and many equation bits jointly score a payload.

This gives two useful noise reductions. First, if each vote for a retained
equation is independently correct with probability \(p_i>1/2\), then majority
voting reduces the bit error probability:
\[
    \Pr[\widehat{y}_i\neq y_i(\payload)]
    \leq
    \exp(-2|V_i|(p_i-1/2)^2).
\]
Second, even when some equation estimates remain wrong, registry scoring uses
many retained equations:
\[
    \mathsf{score}(X,\payload)
    =
    \frac{1}{|A(X)|}
    \sum_{i\in A(X)}
    \mathbf{1}[\widehat{y}_i(X)=y_i(\payload)].
\]
The detector therefore tolerates a nonzero fraction of incorrect equations as
long as the signed payload receives a score above threshold.


Random equations spread evidence across payload bits, while systematic
equations can be used as an ablation in which some equations directly reveal
individual payload bits. In the default construction, we use random low-degree
equations because they create many redundant, partially independent tests of
the same authorized payload.

\subsection{Watermark Authentication}
\label{sec:method:registry}

\scheme{} separates \emph{watermark evidence} from \emph{watermark authority}.
The evidence is the noisy vote pattern extracted from a text. The authority is
a signed registry record saying which payloads are valid for a public context.
The issuer samples a short payload
\[
    \payload \in \{0,1\}^{\ell}
\]
and signs a record binding it to the context and keys:
\[
\begin{aligned}
    R &=
    (\payload,\ell,H(\ctx),\mathsf{fp}(\pk_{\mathsf{sig}}),
      \mathsf{fp}(\pk_{\mathsf{samp}}),\sigma),\\
    \sigma &\leftarrow
    \Sign_{\sk_{\mathsf{sig}}}
    (\payload,\ell,H(\ctx),
      \mathsf{fp}(\pk_{\mathsf{sig}}),
      \mathsf{fp}(\pk_{\mathsf{samp}})).
\end{aligned}
\]
The registry stores \(R\). The payload is intentionally short: it is an
authorized identifier, not a long signature that must be recovered from the
text. Authenticity comes from the signed registry record; the text only needs
to provide evidence for one of the authorized payloads.

For a candidate text \(X\), the verifier accepts only if there exists a
registry record \(R\) such that
\[
\begin{aligned}
    \Verify_{\pk_{\mathsf{sig}}}(R,\ctx)&=1,\\
    \mathsf{score}(X,\payload_R)&\geq\theta,\\
    |A(X)|&\geq n_{\min},\\
    \mathsf{rank}(G_{A(X)})&\geq r_{\min}.
\end{aligned}
\]
The signature prevents an adversary from introducing an arbitrary payload after
seeing a text. The evidence gates prevent acceptance from a small number of
duplicate or low-rank equations. Under EUF-CMA security of the signature
scheme, any successful acceptance for a payload-context pair outside the
registry either yields a signature forgery or a statistical false positive of
the voting test.

\subsection{Security Analysis and Design Discussion}
\label{sec:method:analysis}

The security objective of \scheme{} is deliberately narrower, and more operational, than proving that every accepting document was wholly produced by the issuer. An accepting decision should establish that the supplied text contains enough independent evidence for an issuer-signed, context-bound payload. The construction targets three linked properties: (i) without the sampling secret, fresh token labels should be computationally unpredictable; (ii) redundant token votes should still reveal the intended equation bits when each retained vote is biased toward correctness; and (iii) a plain or unauthorized text should not accumulate enough agreement with any signed registry record to pass the detector, except with the stated statistical tail probability. The signed registry supplies payload authorization, whereas the equation, rank, and score gates supply robustness to noisy textual evidence. 

These guarantees are conditional on the public-key sampling interface described in Section~\ref{subsec:public_verification}. In particular, the current prototype uses a public-extraction placeholder to evaluate sampling, voting, and registry scoring; it does not yet instantiate a concrete
publicly-checkable secret-label primitive. Thus, the theorem below is a registry-backed soundness statement under the DH-label and independent-vote
assumptions, not a claim that a detector-aware adversary with unrestricted access to a fully public predicate cannot optimize an accepting sequence. A copied sufficiently long watermarked span may also remain valid. These boundaries are inherent to the evidence semantics and are discussed after the formal results.

We now prove the three local claims that support this goal, and then state the combined soundness theorem. 

\begin{lemma}[DH label pseudorandomness]
\label{lem:dh-prf}
Assume \(H_{\mathbb{G}}\) is modeled as a random oracle into a prime-order group
\(\mathbb{G}\), DDH holds in \(\mathbb{G}\), and
\(\mathsf{Hash}_{\mathsf{lab}}\) is a pseudorandom extractor for group elements.
For any polynomial adversary that does not know \(\sk_{\mathsf{samp}}\), the
label \(\Label_{\sk_{\mathsf{samp}}}(\tau)\) for a fresh transcript item
\(\tau\) is computationally indistinguishable from
\((I,B)\leftarrow [N]\times\{0,1\}\), up to the DDH and hash advantages.
\end{lemma}

\begin{proof}
Let \(U_\tau=H_{\mathbb{G}}(\tau)\). In the random-oracle model,
\(U_\tau\) is distributed as an independent group element, which we may write
as \(g^r\) for random \(r\in\mathbb{Z}_q\). The real DH label is derived from
\[
    z_\tau=U_\tau^a=g^{ar},
\]
where \(A=g^a\) is public and \(a\) is the sampling secret. By the DDH
assumption, the tuple \((g,g^a,g^r,g^{ar})\) is computationally
indistinguishable from \((g,g^a,g^r,g^u)\) for random
\(u\in\mathbb{Z}_q\). Replacing \(z_\tau\) by a random group element therefore
changes the adversary's view only by the DDH advantage. Applying
\(\mathsf{Hash}_{\mathsf{lab}}\) to an indistinguishable random group element
gives an output indistinguishable from a uniform element of
\([N]\times\{0,1\}\), up to the extractor/hash advantage.
\end{proof}

\begin{lemma}[Rejection-sampling distortion]
\label{lem:distortion}
At position \(t\), let \(M_t(\payload)\subseteq C_t\) be the compatible
candidate set and
\(\rho_t=\widetilde{P}_t(M_t(\payload))\). If
\(\rho_t>0\) and the generator samples
\(Q_t(x)=\widetilde{P}_t(x\mid x\in M_t(\payload))\), then
\[
    D_{\mathrm{KL}}(Q_t\|\widetilde{P}_t)
    =
    \log \frac{1}{\rho_t}.
\]
\end{lemma}

\begin{proof}
For \(x\in M_t(\payload)\),
\[
    Q_t(x)=\frac{\widetilde{P}_t(x)}{\rho_t},
    \qquad
    \rho_t=\sum_{x\in M_t(\payload)}\widetilde{P}_t(x),
\]
and \(Q_t(x)=0\) otherwise. Thus
\[
\begin{aligned}
    D_{\mathrm{KL}}(Q_t\|\widetilde{P}_t)
    &=
    \sum_{x\in M_t(\payload)}
    Q_t(x)\log\frac{Q_t(x)}{\widetilde{P}_t(x)}\\
    &=
    \sum_{x\in M_t(\payload)}
    Q_t(x)\log\frac{1}{\rho_t}
    =
    \log\frac{1}{\rho_t}.
\end{aligned}
\]
which proves the lemma.
\end{proof}

\begin{lemma}[Majority-vote concentration]
\label{lem:majority}
For equation \(i\), suppose the verifier obtains \(r\) independent votes, each
equal to \(y_i(\payload)\) with probability at least \(p>1/2\). Then the
majority estimate satisfies
\[
    \Pr[\widehat{y}_i\neq y_i(\payload)]
    \leq
    \exp(-2r(p-1/2)^2).
\]
\end{lemma}

\begin{proof}
Let \(Z_j\in\{0,1\}\) indicate whether the \(j\)-th vote for equation \(i\) is
correct. By assumption, the \(Z_j\)'s are independent and
\(\mathbb{E}[Z_j]\geq p>1/2\). The majority estimate is wrong only if
\(\sum_{j=1}^{r} Z_j\leq r/2\). Hoeffding's inequality gives
\[
    \Pr\!\left[
        \frac{1}{r}\sum_{j=1}^{r}Z_j - p \leq -(p-1/2)
    \right]
    \leq
    \exp(-2r(p-1/2)^2),
\]
which proves the stated bound.
\end{proof}

\begin{theorem}[Registry-backed soundness]
\label{thm:soundness}
Suppose the signature scheme is EUF-CMA secure, Lemma~\ref{lem:dh-prf} holds,
and for unwatermarked text every retained equation agrees with any fixed
authorized payload with probability at most \(1/2+\epsilon\). For threshold
\(\theta>1/2+\epsilon\) and \(n=|A(X)|\) retained equations, the probability
that an adversary makes \(\mathsf{Detect}\) accept an unauthorized or
unwatermarked text is bounded by
\[
\begin{aligned}
    \mathsf{Adv}_{\Pi}^{\mathsf{sound}}(\Adv)
    \leq\;&
    \mathsf{Adv}_{\Sigma}^{\mathsf{euf}}(\mathcal{B})
    + \mathsf{Adv}_{\mathbb{G}}^{\mathsf{ddh}}(\mathcal{C})\\
    &+ |\Reg|\exp(-2n(\theta-1/2-\epsilon)^2).
\end{aligned}
\]
\end{theorem}
\begin{proof}
	Consider an adversary that causes \(\mathsf{Detect}\) to accept. The detector
	only considers records whose signatures and context bindings verify. If the
	accepted record is not a valid issuer-signed record for the public context and
	keys, the adversary supplies a valid EUF-CMA forgery. This event is bounded by
	\(\mathsf{Adv}_{\Sigma}^{\mathsf{euf}}(\mathcal{B})\).
	
	Otherwise, the accepted record is a genuinely signed registry record. For an
	unwatermarked or unauthorized text, apply the hybrid of
	Lemma~\ref{lem:dh-prf} to replace the retained DH-derived labels with
	pseudorandom labels. The change in acceptance probability is bounded by
	\(\mathsf{Adv}_{\mathbb{G}}^{\mathsf{ddh}}(\mathcal{C})\), together with the
	hash/extractor loss already absorbed by the lemma's assumption. In the hybrid
	experiment, each of the \(n\) retained independent equations agrees with any
	fixed authorized payload with probability at most \(1/2+\epsilon\).
	
	For a fixed signed record and a threshold
	\(\theta>1/2+\epsilon\), Hoeffding's inequality gives
	\[
	\Pr[\mathsf{score}(X,\payload_R)\geq\theta]
	\leq
	\exp(-2n(\theta-1/2-\epsilon)^2).
	\]
	The detector may test at most \(|\Reg|\) signed records, so a union bound
	multiplies this probability by \(|\Reg|\). Adding the signature-forgery and
	DH-hybrid events gives
	\[
	\begin{aligned}
		\mathsf{Adv}_{\Pi}^{\mathsf{sound}}(\Adv)
		\leq\;&
		\mathsf{Adv}_{\Sigma}^{\mathsf{euf}}(\mathcal{B})
		+ \mathsf{Adv}_{\mathbb{G}}^{\mathsf{ddh}}(\mathcal{C})\\
		&+ |\Reg|\exp(-2n(\theta-1/2-\epsilon)^2),
	\end{aligned}
	\]
	which proves the theorem.
\end{proof}

\paragraph{Why registry-aided scoring?}
Exact payload decoding is brittle: a fallback token, local edit, or
tokenization shift can corrupt a packet-like encoding. \scheme{} instead treats
the text as a noisy evidence source. The registry gives the detector a bounded
set of authorized payload hypotheses, and the token sequence supplies votes for
or against each hypothesis. This changes the recovery target from ``decode the
payload exactly'' to ``accumulate sufficient independent evidence for a signed
payload,'' which better matches public provenance.

\paragraph{Quality and fallback.}
Lemma~\ref{lem:distortion} also explains the quality tradeoff. If the compatible
candidate mass \(\rho_t\) is large, conditioning introduces little local
distortion. If \(\rho_t\) is small, insisting on a compatible token would force
low-quality choices. The fallback rule sacrifices one noisy vote at such
positions to keep generation close to the deployed top-\(k\)/top-\(p\) sampler.

\paragraph{Copy-paste and document integrity.}
\scheme{} verifies the presence of sufficient issuer-authorized watermark
evidence. It is not a proof that every token in the document is untampered. If
an adversary pastes a valid watermarked span into a larger document, the span
may still verify. Stronger document-integrity claims require signed text
commitments, span localization, or application-specific policies layered above
the watermark detector.

\section{Evaluation}
\label{sec:evaluation}

We evaluate \scheme{} as a public-key and registry-aided watermarking mechanism for LLM-generated text. Our experiments examine not only whether the final detector accepts or rejects a text, but also whether the extracted token votes provide sufficient equation coverage, linear independence, and agreement with an issuer-authorized registry payload. In particular, we study the following research questions. \\

\begin{compactitem}
	\item[\textbf{RQ1: Detectability and soundness.}] Can a public verifier reliably detect watermarked generations while rejecting ordinary generations and watermarked texts verified under an incorrect public context?
	\item[\textbf{RQ2: Robustness.}] How much valid watermark evidence remains after token deletion, substitution, truncation, copy-paste, and suffix-injection attacks?
	\item[\textbf{RQ3: Parameter sensitivity.}] How do the generated text length, payload length, registry-score threshold, and local context length affect robustness and false-positive behavior?
	\item[\textbf{RQ4: Generation overhead.}] What sampling overhead is introduced by watermark embedding, and how often does the generator fall back to ordinary sampling?
\end{compactitem}

\subsection{Experimental Setup}
\label{subsec:experimental_setup}

Unless otherwise specified, we use an open-weight LLaMA-3 8B model as the underlying generator. We apply top-$k$ and top-$p$ filtering with

$$
k=200,
\qquad
p=0.95,
$$

and use a sampling temperature of $1.0$. Each default run generates $2{,}000$ new tokens. The default watermark payload contains $32$ bits and is expanded into $96$ public one-bit equations with equation degree $3$. The local transcript used by the token-labeling function contains the previous four tokens.

The registry-aided detector uses a registry-agreement threshold of $0.5$. It additionally requires at least $32$ accepted equations, at least $32$ unique equation indices, and a minimum equation rank of $24$. Unless otherwise stated, the default configuration is therefore

$$
(\ell,N,d,h,T,\theta)
=
(32,96,3,4,2000,0.5),
$$

where $\ell$ is the payload length, $N$ is the number of one-bit equations, $d$ is the equation degree, $h$ is the local context length, $T$ is the generation length, and $\theta$ is the registry-score threshold.

We use $30$ prompts for each evaluated configuration. The same prompt collection and attack suite are used within each parameter sweep. For every watermarked continuation, we additionally generate a plain continuation from the same prompt as a negative control. Consequently, the reported changes between configurations mainly reflect the evaluated watermark parameters rather than changes in the prompt distribution.

The current evaluation is intended to characterize the operating region and parameter tradeoffs of the prototype. With $30$ samples, one failure changes the reported valid rate by

$$
\frac{1}{30}\approx 0.0333.
$$

Therefore, differences such as $0.967$ versus $1.000$ correspond to only one sample and should not be interpreted as statistically definitive rankings between closely performing configurations. The consistent zero acceptance rate of the negative controls is nevertheless an important empirical result, although substantially larger negative corpora are required to estimate very small false-positive probabilities.

\subsection{Attacks and Negative Controls}
\label{subsec:attacks}

We evaluate the following transformations.

\textbf{Identity.}
The detector verifies the original watermarked generation without modification. This setting measures basic watermark detectability.

\textbf{Burst deletion.}
A contiguous span of $128$ tokens is removed. This attack tests whether the local transcript and equation-voting mechanism can recover after a localized synchronization loss.

\textbf{Middle crop.}
A contiguous middle segment containing $25\%$ of the generated tokens is removed. Compared with burst deletion, this transformation removes a larger fraction of watermark evidence while preserving both the beginning and the end of the sequence.

\textbf{Random deletion.}
Ten percent of the generated tokens are independently removed. Random deletion causes repeated local transcript disruption throughout the text and is therefore more challenging than deleting a single contiguous suffix.

\textbf{Random substitution.}
Ten percent of the tokens are replaced. This attack both destroys original watermark votes and introduces new votes that are generally unrelated to the registered payload.

\textbf{Prefix and suffix truncation.}
Twenty-five percent of the generated tokens are removed from the beginning or end. Prefix truncation is generally more disruptive because it changes the local transcript reconstructed for the earliest retained tokens, whereas suffix truncation mainly reduces the available amount of evidence.

\textbf{Copy-paste span.}
Only a subspan of the watermarked generation is retained. This test evaluates whether partial provenance evidence remains detectable when a watermarked passage is copied into another document.

\textbf{Malicious suffix.}
Unrelated text is appended to the watermarked output. The appended text contributes noisy or irrelevant votes and tests whether valid evidence can survive dilution by an unwatermarked suffix.

We also include three negative controls. \emph{Wrong public context} verifies watermarked text using a mismatched public context. \emph{Plain generation} verifies an ordinary unwatermarked continuation under the nominal context. \emph{Plain wrong context} combines ordinary generation with a mismatched context. These controls test whether the detector accepts texts merely because they contain many publicly extractable token labels, rather than because those labels agree with an authorized registry payload.

\subsection{Evaluation Metrics}
\label{subsec:metrics}

Our primary metric is the \emph{valid rate},

$$
\mathrm{ValidRate}
=
\frac{1}{n}
\sum_{j=1}^{n}
\mathbf{1}
\left[
\mathsf{Detect}(X_j)=1
\right],
$$

where $n=30$ is the number of evaluated samples in each configuration.

For each text, the detector groups public token votes by equation index. We report the average number of equations satisfying the evidence gates, denoted by \emph{accepted equations}. We also report the average rank of the retained equation matrix over $\mathbb{F}_2$ and the number of unique equation indices receiving sufficient evidence. These measurements distinguish three different conditions:

\begin{compactenum}
	\item whether the text contains many extractable token votes;
	\item whether the votes cover sufficiently many different equations;
	\item whether the resulting equation system agrees with a signed registry payload.
\end{compactenum}

For generation overhead, we use the embedding success rate

$$
r_{\mathrm{succ}}
=
\frac{T_{\mathrm{matched}}}{T},
$$

and fallback rate

$$
r_{\mathrm{fb}}
=
\frac{T_{\mathrm{fallback}}}{T}
=
1-r_{\mathrm{succ}},
$$

where $T_{\mathrm{matched}}$ is the number of positions at which a watermark-compatible token is selected and $T_{\mathrm{fallback}}$ is the number of positions generated using the quality-preserving fallback rule. We additionally report the average matching probability mass, candidate entropy, number of matching candidates, rejection attempts per token, and elapsed generation time when analyzing sampling behavior.

\begin{table*}[t]
	\footnotesize
	\setlength{\abovecaptionskip}{0cm}
	\centering
	\caption{Detection results of the default configuration on $30$ prompts. The registry-score threshold is $0.5$.}
	\tabcolsep=10pt
	\vspace{10pt}
	\begin{tabular}{lcccc}
		\hline \hline \noalign{\smallskip}
		Attack & Valid rate & Avg. accepted eq. & Avg. rank & Avg. unique eq. \\
		\midrule
		Identity                  & 1.000 & 93.83 & 32.00 & 95.93 \\
		Burst deletion, 128       & 1.000 & 93.63 & 32.00 & 95.93 \\
		Middle crop, 25\%         & 1.000 & 93.13 & 32.00 & 95.93 \\
		Random deletion, 10\%     & 0.967 & 91.53 & 32.00 & 96.00 \\
		Random substitution, 10\% & 0.967 & 91.57 & 32.00 & 95.97 \\
		Prefix truncation, 25\%   & 0.967 & 89.43 & 31.97 & 93.07 \\
		Suffix truncation, 25\%   & 1.000 & 93.20 & 32.00 & 95.93 \\
		Copy-paste span           & 1.000 & 90.97 & 32.00 & 95.87 \\
		Malicious suffix          & 1.000 & 93.80 & 32.00 & 96.00 \\
		\midrule
		Wrong public context      & 0.000 & 88.07 & 32.00 & 95.97 \\
		Plain generation          & 0.000 & 84.77 & 32.00 & 94.43 \\
		Plain wrong context       & 0.000 & 84.87 & 32.00 & 94.70 \\
		\noalign{\smallskip} \hline \hline
	\end{tabular}
	\label{tab:main_results}
\end{table*}

\begin{figure*}[t]
	\centering
\begin{tikzpicture}
\begin{groupplot}[
	group style={
		group size=2 by 1,
		horizontal sep=1.05cm,
	},
	width=0.455\textwidth,
	height=0.39\textwidth,
	xbar,
	/pgf/bar width=4.1pt,
	symbolic y coords={
		Identity,
		Burst deletion,
		Middle crop,
		Random deletion,
		Random substitution,
		Prefix truncation,
		Suffix truncation,
		Copy-paste,
		Malicious suffix,
		Wrong context,
		Plain generation,
		Plain wrong context
	},
	ytick={
		Identity,
		Burst deletion,
		Middle crop,
		Random deletion,
		Random substitution,
		Prefix truncation,
		Suffix truncation,
		Copy-paste,
		Malicious suffix,
		Wrong context,
		Plain generation,
		Plain wrong context
	},
	y dir=reverse,
	axis x line*=bottom,
	axis y line*=left,
	tick align=outside,
	tick style={black,thin},
	xmajorgrids,
	grid style={plotgrid,thin},
	nodes near coords,
	every node near coord/.append style={
		font=\scriptsize,
		anchor=west,
		xshift=1pt,
	},
	title style={font=\small\bfseries,align=left},
	label style={font=\small},
	tick label style={font=\scriptsize},
	clip=false,
]

\nextgroupplot[
	title={(a) Detection under edits and negative controls},
	title style={font=\normalfont, color=black, xshift=-8pt,yshift=-2pt},
	xmin=0,
	xmax=1.05,
	xtick={0,0.2,0.4,0.6,0.8,1.0},
	xticklabel style={/pgf/number format/fixed,/pgf/number format/precision=2},
	xlabel={Valid rate},
	point meta=x,
	nodes near coords={\pgfmathprintnumber[fixed,zerofill,precision=3]{\pgfplotspointmeta}},
]
\addplot[fill=plotteal,draw=none] coordinates {
	(1.000,Identity)
	(1.000,Burst deletion)
	(1.000,Middle crop)
	(0.967,Random deletion)
	(0.967,Random substitution)
	(0.967,Prefix truncation)
	(1.000,Suffix truncation)
	(1.000,Copy-paste)
	(1.000,Malicious suffix)
	(0.000,Wrong context)
	(0.000,Plain generation)
	(0.000,Plain wrong context)
};
\draw[plotamber,dashed,thick]
	(axis cs:0.95,Identity) -- (axis cs:0.95,Plain wrong context);

\nextgroupplot[
	title={(b) Evidence remains abundant in negatives},
	title style={font=\normalfont, color=black, xshift=-8pt,yshift=-2pt},
	xmin=80,
	xmax=97.5,
	xtick={80,84,88,92,96},
	xlabel={Average accepted equations (of 96)},
	yticklabels={},
	point meta=x,
	nodes near coords={\pgfmathprintnumber[fixed,precision=1]{\pgfplotspointmeta}},
]
\addplot[fill=plotteal,draw=none] coordinates {
	(93.83,Identity)
	(93.63,Burst deletion)
	(93.13,Middle crop)
	(91.53,Random deletion)
	(91.57,Random substitution)
	(89.43,Prefix truncation)
	(93.20,Suffix truncation)
	(90.97,Copy-paste)
	(93.80,Malicious suffix)
	(88.07,Wrong context)
	(84.77,Plain generation)
	(84.87,Plain wrong context)
};

\end{groupplot}
\end{tikzpicture}
	\caption{Default-configuration robustness and evidence coverage. Although
		negative texts still produce many accepted equations, their registry valid
		rate is zero; equation availability alone is therefore not an authentication
		signal. The amber dashed line marks a valid rate of $0.95$.}
	\label{fig:robustness-summary}
\end{figure*}
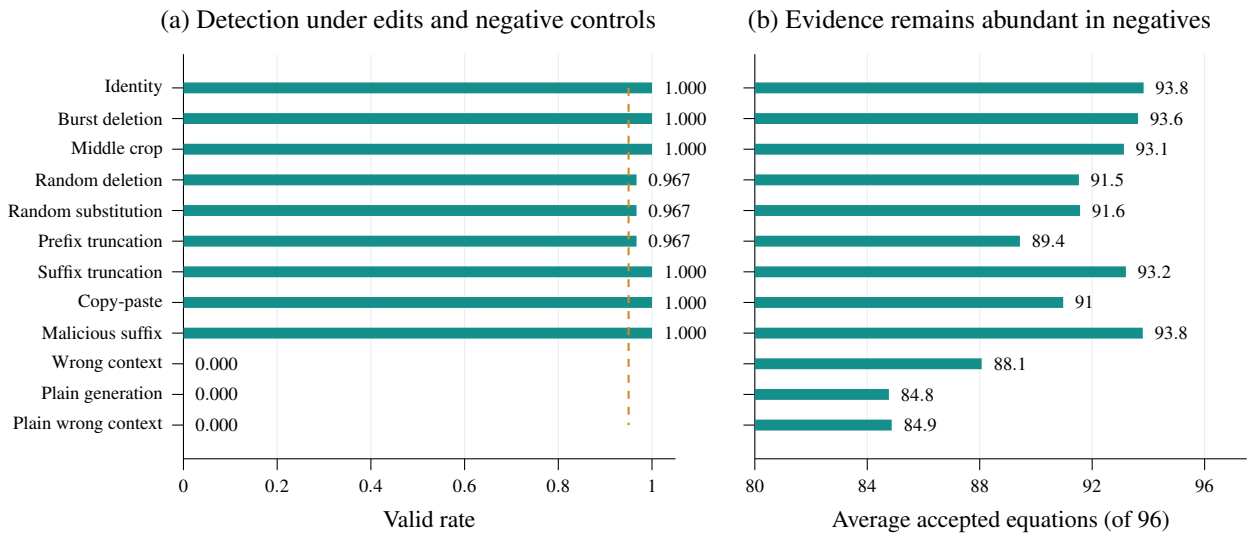

\subsection{Main Detection and Robustness Results}
\label{subsec:main_results}

Table~\ref{tab:main_results} reports the default configuration with a registry-score threshold of $0.5$. \scheme{} accepts all unmodified watermarked texts. It also achieves a valid rate of $1.000$ under burst deletion, middle cropping, suffix truncation, copy-paste, and malicious suffix injection. Random deletion, random substitution, and prefix truncation each cause one failure among the $30$ samples, corresponding to a valid rate of $0.967$.

Most importantly, none of the three negative controls are accepted. The wrong-context result shows that a valid watermark signal is bound to the public context rather than only to token-level statistical bias. Plain generations are also rejected even though they yield many public labels and almost complete equation coverage.

The accepted-equation statistics reveal an important property of the proposed detector. Plain and wrong-context texts still produce approximately $85$--$88$ accepted equations, ranks close to $32$, and more than $94$ unique equation indices. Therefore, equation availability and full rank alone are not sufficient for authentication. The decisive signal is whether the majority estimates of these equations agree with one particular signed registry payload above the selected threshold.

This result supports the separation between \emph{public extraction} and \emph{registry-backed authentication}. Public extraction ensures that any verifier can reconstruct token-level evidence. Registry scoring determines whether this evidence is statistically aligned with an issuer-authorized payload.

Among the attacks, prefix truncation results in the lowest average equation coverage. Its average number of unique equations decreases to $93.07$, compared with $95.93$ for identity. This is consistent with the use of a local transcript: deleting the prefix alters the transcript reconstructed at the beginning of the retained sequence and temporarily disrupts label synchronization. In contrast, suffix truncation preserves all preceding local transcripts and primarily removes evidence from the end of the sequence, resulting in no observed detection failures.

\begin{table*}[t]
	\footnotesize
	\setlength{\abovecaptionskip}{0cm}
	\centering
	\caption{Sensitivity to the registry-score threshold. "Positive average" is averaged over the nine watermarked conditions, while "negative maximum" is the maximum valid rate among the three negative controls. }
	\tabcolsep=10pt
	\vspace{10pt}
	\begin{tabular}{ccccc}
		\hline \hline \noalign{\smallskip}
		Threshold & Identity & Copy-paste & Positive average & Negative maximum \\
		\midrule
		0.3 & 1.000 & 1.000 & 1.000 & 0.000 \\
		0.5 & 1.000 & 1.000 & 0.989 & 0.000 \\
		0.7 & 0.967 & 0.933 & 0.941 & 0.000 \\
		\noalign{\smallskip} \hline \hline
	\end{tabular}
	\label{tab:threshold_results}
\end{table*}

\begin{table*}[t]
	\footnotesize
	\setlength{\abovecaptionskip}{0cm}
	\centering
	\caption{Sensitivity to generated text length.}
	\tabcolsep=10pt
	\vspace{10pt}
	\begin{tabular}{rccccc}
		\hline \hline \noalign{\smallskip}
		Tokens & Identity & Positive average & Positive minimum & Negative maximum & Attempts/token \\
		\midrule
		500   & 1.000 & 0.952 & 0.800 & 0.000 & 4.14 \\
		1,000 & 1.000 & 0.996 & 0.967 & 0.000 & 3.88 \\
		1,500 & 1.000 & 0.989 & 0.967 & 0.000 & 3.29 \\
		2,000 & 1.000 & 0.989 & 0.967 & 0.000 & 3.00 \\
		3,000 & 1.000 & 0.959 & 0.867 & 0.000 & 2.67 \\
		\noalign{\smallskip} \hline \hline
	\end{tabular}
	\label{tab:length_results}
\end{table*}

Random deletion and random substitution are also more difficult than a single burst deletion. A localized burst affects a bounded region and the following context window, after which the detector can resynchronize. Random edits repeatedly disrupt local contexts throughout the sequence. Nevertheless, the redundant equation construction retains enough correct evidence to accept $29$ of the $30$ samples under both transformations.

Figure~\ref{fig:robustness-summary} makes the separation between evidence
availability and authentication explicit. The nine watermarked conditions
cluster between $0.967$ and $1.000$ valid rate, whereas all three negative
controls remain at zero even though they retain roughly $85$--$88$ accepted
equations. The registry agreement test, rather than equation count alone, is
therefore responsible for the observed negative-control rejection.

\subsection{Registry-Score Threshold}
\label{subsec:threshold}

The registry-score threshold determines how strongly the extracted equation estimates must agree with a signed payload. For a candidate record $R$, the detector computes

$$
\mathsf{score}(X,m_R)
=
\frac{1}{|A(X)|}
\sum_{i\in A(X)}
\mathbf{1}
\left[
\widehat{y}_i(X)=y_i(m_R)
\right],
$$
and accepts only when

$$
\mathsf{score}(X,m_R)\geq\theta
$$
and the additional evidence gates are satisfied.

Table~\ref{tab:threshold_results} compares thresholds $0.3$, $0.5$, and $0.7$. A threshold of $0.3$ accepts every attacked watermarked sample in this experiment, while all negative controls remain rejected. However, because random agreement with a fixed binary equation vector is expected to be centered near $0.5$, a threshold below $0.5$ provides a weak conceptual and statistical margin. Its apparent success on the current negative set should therefore not be interpreted as evidence that $0.3$ is generally safe.

Increasing the threshold to $0.7$ makes the detector substantially more conservative. The average positive valid rate decreases from $1.000$ at threshold $0.3$ to $0.941$, and the worst-case valid rate decreases to $0.867$. Even identity detection decreases to $0.967$, showing that the stronger threshold rejects some naturally noisy watermarked outputs.

We select $\theta=0.5$ as a practical operating point. It rejects all negative controls, retains perfect identity detection, and achieves an average valid rate of $0.989$ over the nine positive conditions.

The observed zero negative valid rate at all three thresholds must be interpreted with the sample size in mind. Across the three negative conditions, each threshold is evaluated on $90$ negative texts. Observing no false positives establishes empirical separation in the evaluated batch, but it does not establish a cryptographic or population-level false-positive rate of zero. A larger-scale negative evaluation is required to estimate the tail of the registry-score distribution.

\begin{table*}[t]
	\footnotesize
	\setlength{\abovecaptionskip}{0cm}
	\centering
	\caption{Sensitivity to payload length at $2{,}000$ generated tokens.}
	\tabcolsep=10pt
	\vspace{10pt}
	\begin{tabular}{rccccc}
		\hline \hline \noalign{\smallskip}
		Payload bits & Equations & Identity & Copy-paste & Positive average & Negative maximum \\
		\midrule
		32  & 96  & 1.000 & 0.967 & 0.989 & 0.000 \\
		64  & 192 & 0.967 & 0.933 & 0.948 & 0.000 \\
		128 & 384 & 0.967 & 0.833 & 0.848 & 0.000 \\
		\noalign{\smallskip} \hline \hline
	\end{tabular}
	\label{tab:payload_results}
\end{table*}

\begin{table*}[t]
	\footnotesize
	\setlength{\abovecaptionskip}{0cm}
	\centering
	\caption{Effect of local transcript length.}
	\tabcolsep=10pt
	\vspace{10pt}
	\begin{tabular}{rcccccc}
		\hline \hline \noalign{\smallskip}
		Context & Identity & Copy-paste & Positive avg. & Positive min. & Negative max. & Fallback rate \\
		\midrule
		0 & 0.900 & 0.733 & 0.815 & 0.700 & 0.000 & 0.182 \\
		1 & 1.000 & 0.967 & 0.996 & 0.967 & 0.000 & 0.215 \\
		2 & 0.967 & 0.967 & 0.974 & 0.967 & 0.000 & 0.245 \\
		4 & 1.000 & 1.000 & 1.000 & 1.000 & 0.000 & 0.275 \\
		8 & 1.000 & 1.000 & 1.000 & 1.000 & 0.000 & 0.271 \\
		\noalign{\smallskip} \hline \hline
	\end{tabular}
	\label{tab:context_results}
\end{table*}

\subsection{Effect of Generation Length}
\label{subsec:length}

Text length directly controls the amount of token-level evidence available to the verifier. Table~\ref{tab:length_results} summarizes generation lengths between $500$ and $3{,}000$ tokens.

At $500$ tokens, identity detection is already perfect, showing that the watermark can be detected without requiring the full default length. However, the shorter text is more vulnerable to distributed edits. Random deletion has a valid rate of $0.800$, while random substitution and malicious suffix injection both obtain $0.900$. Since a fixed-size attack removes or corrupts a larger fraction of the available evidence in a shorter sequence, fewer reliable votes remain after aggregation.

Performance improves sharply at $1{,}000$ tokens. The average positive valid rate reaches $0.996$, with the only observed failure occurring under random substitution. The $1{,}500$- and $2{,}000$-token configurations also remain highly robust, both obtaining an average positive valid rate of $0.989$.

Interestingly, increasing the length to $3{,}000$ tokens does not further improve robustness. Copy-paste decreases to $0.867$, random deletion decreases to $0.900$, and the average positive valid rate becomes $0.959$. At the same time, rejection attempts per token decrease from $4.14$ at $500$ tokens to $2.67$ at $3{,}000$ tokens. This suggests that the degradation is not caused by insufficient candidate matches. A plausible explanation is that long autoregressive continuations become more repetitive or locally concentrated, so additional tokens do not necessarily provide proportionally more independent watermark evidence. Because the current CSV does not contain a direct repetition metric, this explanation should be regarded as a hypothesis to be verified through a text-quality analysis rather than as a confirmed conclusion.

These results identify a practical operating region between $1{,}000$ and $2{,}000$ generated tokens. We use $2{,}000$ tokens in the default configuration because it provides a large evidence budget while retaining high robustness and lower average sampling overhead than the shorter settings.

\subsection{Effect of Payload Length}
\label{subsec:payload}

The payload-length experiment evaluates payloads of $32$, $64$, and $128$ bits. The number of equations is scaled proportionally:

$$
N=3\ell,
$$

resulting in $96$, $192$, and $384$ equations, respectively.

Table~\ref{tab:payload_results} shows that increasing the payload length consistently reduces robustness. The $32$-bit configuration achieves perfect identity detection and an average positive valid rate of $0.989$. With $64$ bits, identity decreases to $0.967$ and the average positive valid rate decreases to $0.948$. With $128$ bits, the average decreases further to $0.848$.

The degradation is especially pronounced for distributed corruption. At $128$ bits, random substitution reaches only $0.600$, random deletion reaches $0.667$, and copy-paste reaches $0.833$. In comparison, localized burst deletion still obtains $0.967$. These results indicate that increasing the number of equations does not fully compensate for the larger payload dimension. The detector must collect enough independent and sufficiently accurate evidence to distinguish one registered $128$-bit payload, while edits reduce or corrupt the votes distributed across the larger equation system.

This experiment supports the central design decision of \scheme{}: the text should carry evidence for a \emph{short registry identifier}, while authenticity is supplied by the digital signature stored in the registry. Directly embedding a long signature or a large metadata record would require substantially more reliable evidence and would reduce robustness under practical edits.

The fact that every negative control remains at $0.000$ for all payload lengths also indicates that the loss in positive valid rate is caused by insufficient or noisy authorized evidence, rather than by an increase in false-positive acceptance.

\subsection{Effect of Local Context Length}
\label{subsec:context_length}

The local transcript length $h$ controls how many previous tokens are included when deriving the label for the current token. Table~\ref{tab:context_results} evaluates

$$
h\in\{0,1,2,4,8\}.
$$

With $h=0$, token labels are independent of the preceding local transcript. This configuration performs substantially worse than the context-dependent variants. Identity detection is only $0.900$, copy-paste detection is $0.733$, and the average positive valid rate is $0.815$. Prefix truncation is the most difficult attack, with a valid rate of $0.700$.

\begin{figure*}[t]
	\centering
\begin{tikzpicture}

\begin{scope}[xshift=1cm,yshift=4.7cm,font=\footnotesize]
	\draw[plotnavy,thick,mark=*,mark options={solid,fill=plotnavy}]
		(0,0) -- (0.55,0);
	\node[anchor=west] at (0.65,0) {Identity};
	\draw[plotteal,thick,mark=square*,mark options={solid,fill=plotteal}]
		(2.25,0) -- (2.80,0);
	\node[anchor=west] at (2.90,0) {Positive average};
	\draw[plotcoral,thick,mark=triangle*,mark options={solid,fill=plotcoral}]
		(5.35,0) -- (5.90,0);
	\node[anchor=west] at (6.00,0) {Positive minimum};
	\draw[plotgray,thick,dashed,mark=diamond*,mark options={solid,fill=plotgray}]
		(8.55,0) -- (9.10,0);
	\node[anchor=west] at (9.20,0) {Negative maximum};
\end{scope}

\begin{groupplot}[
	group style={
		group size=2 by 2,
		horizontal sep=1.15cm,
		vertical sep=2.0cm,
	},
	width=0.455\textwidth,
	height=0.26\textwidth,
	xminorticks=false,
	ymin=0.58,
	ymax=1.04,
	ytick={0.6,0.8,1.0},
	ymajorgrids,
	grid style={plotgrid,thin},
	axis x line*=bottom,
	axis y line*=left,
	tick align=outside,
	tick style={black,thin},
	label style={font=\small},
	tick label style={font=\scriptsize},
	clip=false,
	every axis plot/.append style={thick},
]

\nextgroupplot[
	title={(a) Registry threshold},
	title style={font=\normalfont\small,color=black, xshift=-8pt,yshift=-2pt},
	xlabel={Threshold $\theta$},
	ylabel={Valid rate},
	xmin=0.28,
	xmax=0.72,
	xtick={0.3,0.5,0.7},
]
\addplot[plotnavy,mark=*,mark options={fill=plotnavy}] coordinates {
	(0.3,1.000) (0.5,1.000) (0.7,0.967)
};
\addplot[plotteal,mark=square*,mark options={fill=plotteal}] coordinates {
	(0.3,1.000) (0.5,0.989) (0.7,0.941)
};
\addplot[plotcoral,mark=triangle*,mark options={fill=plotcoral}] coordinates {
	(0.3,1.000) (0.5,0.967) (0.7,0.867)
};
\addplot[plotgray,dashed,mark=diamond*,mark options={fill=plotgray}] coordinates {
	(0.3,0.58) (0.5,0.58) (0.7,0.58)
};
\draw[plotamber,dotted,thick]
	(axis cs:0.5,0.58) -- (axis cs:0.5,1.03);

\nextgroupplot[
	title={(b) Generated length},
	title style={font=\normalfont\small,color=black, xshift=-8pt,yshift=-2pt},
	xlabel={Generated tokens},
	xmin=400,
	xmax=3100,
	xtick={500,1000,1500,2000,2500,3000},
	yticklabels={},
]
\addplot[plotnavy,mark=*,mark options={fill=plotnavy}] coordinates {
	(500,1.000) (1000,1.000) (1500,1.000) (2000,1.000) (3000,1.000)
};
\addplot[plotteal,mark=square*,mark options={fill=plotteal}] coordinates {
	(500,0.952) (1000,0.996) (1500,0.989) (2000,0.989) (3000,0.959)
};
\addplot[plotcoral,mark=triangle*,mark options={fill=plotcoral}] coordinates {
	(500,0.800) (1000,0.967) (1500,0.967) (2000,0.967) (3000,0.867)
};
\addplot[plotgray,dashed,mark=diamond*,mark options={fill=plotgray}] coordinates {
	(500,0.58) (1000,0.58) (1500,0.58) (2000,0.58) (3000,0.58)
};
\draw[plotamber,dotted,thick]
	(axis cs:2000,0.58) -- (axis cs:2000,1.03);

\nextgroupplot[
	title={(c) Payload length},
	title style={font=\normalfont\small,color=black, xshift=-8pt,yshift=-2pt},
	xlabel={Payload bits},
	ylabel={Valid rate},
	xmin=28,
	xmax=132,
	xtick={32,64,96,128},
]
\addplot[plotnavy,mark=*,mark options={fill=plotnavy}] coordinates {
	(32,1.000) (64,0.967) (128,0.967)
};
\addplot[plotteal,mark=square*,mark options={fill=plotteal}] coordinates {
	(32,0.989) (64,0.948) (128,0.848)
};
\addplot[plotcoral,mark=triangle*,mark options={fill=plotcoral}] coordinates {
	(32,0.967) (64,0.900) (128,0.600)
};
\addplot[plotgray,dashed,mark=diamond*,mark options={fill=plotgray}] coordinates {
	(32,0.58) (64,0.58) (128,0.58)
};
\draw[plotamber,dotted,thick]
	(axis cs:32,0.58) -- (axis cs:32,1.03);

\nextgroupplot[
	title={(d) Local context},
	title style={font=\normalfont\small,color=black, xshift=-8pt,yshift=-2pt},
	xlabel={Context length $h$},
	xmin=-0.35,
	xmax=8.35,
	xtick={0,1,2,4,8},
	yticklabels={},
]
\addplot[plotnavy,mark=*,mark options={fill=plotnavy}] coordinates {
	(0,0.900) (1,1.000) (2,0.967) (4,1.000) (8,1.000)
};
\addplot[plotteal,mark=square*,mark options={fill=plotteal}] coordinates {
	(0,0.815) (1,0.996) (2,0.974) (4,1.000) (8,1.000)
};
\addplot[plotcoral,mark=triangle*,mark options={fill=plotcoral}] coordinates {
	(0,0.700) (1,0.967) (2,0.967) (4,1.000) (8,1.000)
};
\addplot[plotgray,dashed,mark=diamond*,mark options={fill=plotgray}] coordinates {
	(0,0.58) (1,0.58) (2,0.58) (4,0.58) (8,0.58)
};
\draw[plotamber,dotted,thick]
	(axis cs:4,0.58) -- (axis cs:4,1.03);

\end{groupplot}
\end{tikzpicture}
	\caption{Parameter sensitivity across the registry threshold, generation
		length, payload length, and local context length. The positive minimum
		highlights the most difficult watermarked condition in each sweep; the
		negative maximum remains zero throughout. Amber dotted lines identify the
		default configuration. Each point summarizes $30$ prompts.}
	\label{fig:parameter-sensitivity}
\end{figure*}
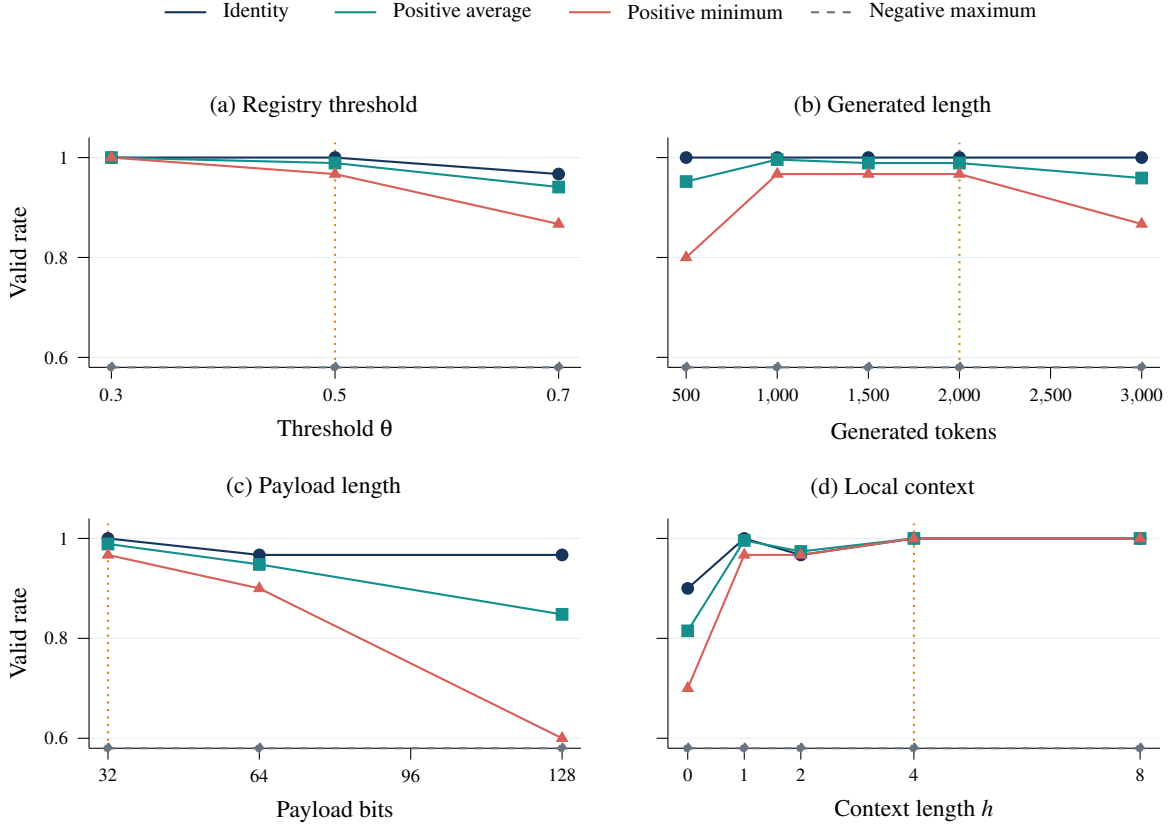

Using only one preceding token raises the average positive valid rate to $0.996$. The two-token setting obtains $0.974$, while context lengths $4$ and $8$ accept every positive sample in their respective batches. All negative controls remain rejected for every context length.

The result for $h=0$ shows that local context is not merely a synchronization cost. Conditioning labels on a local transcript increases the diversity of token-label inputs and reduces repeated use of identical token labels in repetitive lexical patterns. This improves equation coverage and rank. For example, under identity verification, the context-free configuration obtains approximately $81.3$ unique equations and rank $31.23$, whereas the context-dependent configurations approach full coverage and rank.

Longer contexts introduce additional generation cost. Rejection attempts per token increase from $2.04$ at $h=0$ to $3.12$ at $h=4$ and $3.28$ at $h=8$. The fallback rate similarly increases from $0.182$ to approximately $0.27$. Since $h=4$ and $h=8$ have identical valid rates in this experiment, we select $h=4$ as the default to obtain the robustness benefit without the additional context-processing and sampling cost of the longer window.

The apparent non-monotonic result at $h=2$, which performs slightly worse than $h=1$, is small in absolute terms and corresponds to a limited number of samples. It should not be interpreted as evidence that a one-token context is intrinsically superior. The more reliable conclusion is that introducing local context provides a major improvement over $h=0$, while lengths between $1$ and $8$ offer a tradeoff between robustness and embedding cost.

Figure~\ref{fig:parameter-sensitivity} consolidates the four sweeps and shows
why the default configuration is a useful operating point. The selected
threshold preserves separation from the negative controls, lengths between
$1{,}000$ and $2{,}000$ avoid the weakest short- and long-generation results,
the $32$-bit payload avoids the robustness loss of larger payloads, and
$h=4$ reaches the best observed positive rates without extending the context
to eight tokens.

\begin{table*}[t]
	\footnotesize
	\setlength{\abovecaptionskip}{0cm}
	\centering
	\caption{Per-attack valid rates in the length sweep. All rows use 32-bit payloads, 96 degree-3 equations, context length 4, and threshold 0.5.}
	\tabcolsep=10pt
	\vspace{10pt}
	\begin{tabular}{rcccccccccc}
		\hline \hline \noalign{\smallskip}
		Tokens & Id. & Burst & Copy & Suffix & Mid & Rand del. & Rand sub. & Pre. trunc. & Suf. trunc. & Neg. max \\
		\midrule
		500 & 1.000 & 1.000 & 0.967 & 0.900 & 1.000 & 0.800 & 0.900 & 1.000 & 1.000 & 0.000 \\
		1,000 & 1.000 & 1.000 & 1.000 & 1.000 & 1.000 & 1.000 & 0.967 & 1.000 & 1.000 & 0.000 \\
		1,500 & 1.000 & 1.000 & 0.967 & 1.000 & 1.000 & 0.967 & 0.967 & 1.000 & 1.000 & 0.000 \\
		2,000 & 1.000 & 1.000 & 0.967 & 1.000 & 1.000 & 0.967 & 1.000 & 0.967 & 1.000 & 0.000 \\
		3,000 & 1.000 & 1.000 & 0.867 & 0.967 & 1.000 & 0.900 & 1.000 & 0.933 & 1.000 & 0.000 \\
		\noalign{\smallskip} \hline \hline
	\end{tabular}
	\label{tab:app-length-attacks}
\end{table*}

\subsection{Generation Overhead}
\label{subsec:generation_cost}

For the default parameter region, approximately three quarters of token positions successfully select watermark-compatible candidates, while the remaining positions use quality-preserving fallback sampling. In the threshold sweep, the embedding success rate ranges from approximately $0.732$ to $0.745$, corresponding to fallback rates between $0.255$ and $0.268$. The default threshold-$0.5$ batch has
$$
r_{\mathrm{succ}}=0.7323
$$
and
$$
r_{\mathrm{fb}}=0.2677.
$$

For a $2{,}000$-token generation, these rates correspond to approximately
$$
2000\times 0.7323 = 1464.6
$$
watermark-compatible positions and
$$
2000\times 0.2677 = 535.4
$$
fallback positions.

The default batch requires $3.24$ rejection attempts per generated token on average. The average matching probability mass is $0.754$, and the average number of matching candidates is approximately $5.81$. These results indicate that compatible candidates generally occupy substantial probability mass, but candidate availability is uneven across generation positions. The fallback mechanism is therefore necessary to avoid forcing low-probability candidates at positions where watermark-compatible choices are unavailable or unsuitable.

The parameter sweeps reveal two additional trends. First, rejection attempts per token decrease as generation length increases, from $4.14$ at $500$ tokens to $2.67$ at $3{,}000$ tokens. Second, increasing the local context length raises the sampling cost: attempts per token increase from $2.04$ at $h=0$ to $3.28$ at $h=8$. These results show that the watermark strength and context binding are obtained at a measurable but bounded sampling cost.

The current data report sampling effort and fallback behavior but do not contain direct text-quality measurements such as perplexity, semantic similarity, repetition rate, or human preference. Therefore, the fallback statistics support the intended quality-preserving behavior of the algorithm, but they do not by themselves establish that generated text quality is unchanged. A complete quality evaluation should additionally compare plain and watermarked outputs using both automated and human-centered metrics.

\subsection{Additional Results}

Table~\ref{tab:app-length-attacks} expands the length sweep by listing the
main positive and negative attacks. The table supports the trend discussed in
Section~4: identity detection remains stable across all lengths, while attacks
that retain only partial or noisy evidence are more sensitive to the usable
token budget and to long-generation degeneration.

Across all final sweeps, the maximum valid rate among the plain-generation,
plain-wrong-context, and wrong-public-context controls is 0.000. This does not
prove a negligible false-positive probability by itself, because the prompt set
contains only 30 prompts per configuration. It nevertheless confirms that the
registry score and context binding do not accidentally accept the negative
controls used in our benchmark.

\subsection{Summary of Findings}
\label{subsec:evaluation_summary}

The experimental results provide the following answers to the research questions.

\textbf{RQ1: Detectability and soundness.}
The default configuration detects all unmodified watermarked texts and rejects all plain-generation and wrong-context controls. Negative texts may still provide high equation coverage and full rank, but they fail registry-agreement verification. This confirms that public token extraction alone is not used as authentication.

\textbf{RQ2: Robustness.}
The watermark remains detectable under all evaluated transformations with valid rates between $0.967$ and $1.000$ in the default experiment. Localized deletion, middle cropping, suffix truncation, copy-paste, and malicious suffix injection cause no observed failures. Distributed deletion, substitution, and prefix truncation each cause one failure among $30$ samples.

\textbf{RQ3: Parameter sensitivity.}
A registry threshold of $0.5$ provides a strong balance between attacked-text robustness and negative-control rejection. Generated lengths between $1{,}000$ and $2{,}000$ tokens form the most reliable operating region. Short payloads are substantially more robust than long payloads, supporting the use of compact signed registry identifiers. Local transcript binding is essential: context lengths $4$ and $8$ outperform context-free labeling, while $h=4$ has lower overhead than $h=8$.

\textbf{RQ4: Generation overhead.}
Approximately $73\%$ of positions successfully embed compatible evidence in the default batch, while approximately $27\%$ use fallback sampling. The generator requires about $3.24$ rejection attempts per token. This overhead is non-negligible but remains bounded, and the fallback mechanism prevents the generator from forcing incompatible low-quality tokens.

Overall, the experiments support the central design of \scheme{}: public watermark verification is more robust when a short authorized payload is represented by many noisy local equations and detected through registry-aided evidence aggregation, rather than through exact recovery of a long cryptographic string.

\section{Discussion}
\label{sec:discussion}

\subsection{Public Verification Semantics}
\label{subsec:verification_semantics}

The verification result produced by \scheme{} should be interpreted carefully. An accepting result indicates that the examined text contains sufficient watermark evidence consistent with an issuer-authorized registry payload under the specified public context. It does not prove that every token in the document was generated by the issuer, nor does it certify the integrity of the complete document.

Formally, the detector establishes that there exists a valid registry record $R$ satisfying
$$
\mathsf{Verify}_{\mathsf{pk}_{\mathsf{sig}}}(R,ctx)=1,
$$
and
$$
\mathsf{score}(X,m_R)\geq\theta.
$$

Therefore, the semantic meaning of a successful verification is \emph{the text contains issuer-authorized watermark evidence}, rather than \emph{the entire document is an authentic issuer-generated document}. 

This distinction is particularly important under copy-paste and malicious-suffix attacks. A copied span from a watermarked document may still contain sufficient evidence for successful verification, even if the surrounding text originates from another source. Likewise, appending unrelated content introduces additional noisy votes but does not necessarily invalidate the existing watermark evidence. Consequently, \scheme{} is designed for provenance verification rather than document authentication. Applications requiring full-document integrity should combine the proposed watermark with complementary mechanisms such as document signatures or authenticated metadata.

\subsection{Why Registry-Aided Evidence Aggregation}
\label{subsec:evidence_aggregation}

A key design decision of \scheme{} is to verify registry-authorized evidence instead of recovering an exact embedded signature. This distinction is motivated by the inherently noisy nature of LLM text generation. Token substitutions, deletions, fallback sampling, and local synchronization errors all introduce uncertainty into the extracted watermark signal. Recovering a complete embedded bitstream therefore becomes increasingly fragile as the embedded payload grows.

Instead, \scheme{} expands a short payload into multiple one-bit equations and evaluates their statistical agreement with signed registry records. Verification succeeds whenever the aggregated evidence exceeds the decision threshold,
$$
\mathsf{score}(X,m_R)\geq\theta,
$$
rather than requiring exact recovery of every embedded bit.

This registry-aided formulation naturally matches the characteristics of LLM-generated text. Individual equation votes may be corrupted by editing or sampling noise, while redundant observations collected throughout the generated sequence still provide sufficient evidence for reliable verification. The experimental results confirm that this strategy remains robust under various token-level transformations while maintaining strong rejection performance on wrong-context and plain-generation inputs.

\subsection{Deployment Considerations}
\label{subsec:deployment}

The experimental results also provide practical guidance for deploying \scheme{}.

First, the payload embedded into the generated text should remain relatively short. Rather than embedding complete metadata or digital signatures, the payload only serves as a registry identifier. Rich metadata—including issuer identity, timestamps, generation policies, or application-specific information—can be stored in the signed registry record. This separation substantially improves robustness while preserving the ability to verify cryptographic authenticity.

Second, public verification should always be interpreted together with the associated registry. The detector verifies whether the observed watermark evidence is consistent with an authorized registry entry under the specified public context. Consequently, the registry becomes an integral component of the verification pipeline instead of merely serving as auxiliary metadata storage. This design naturally supports offline third-party verification without requiring access to the issuer's private detection service.

\subsection{Limitations and Future Work}
\label{subsec:limitations}

Although the proposed framework demonstrates encouraging robustness and public verifiability, several limitations remain.

First, the current evaluation primarily considers token-level editing operations, including deletion, substitution, truncation, copy-paste, and suffix injection. More semantic transformations, such as paraphrasing, summarization, and machine translation, may significantly alter local token distributions and deserve further investigation.

Second, the current implementation requires repeated candidate evaluation during generation, introducing additional sampling overhead. More efficient candidate indexing, batched cryptographic computation, and GPU-oriented implementations could substantially improve generation efficiency.

Finally, the present work focuses on registry-backed public verification rather than the strongest notion of generation unforgeability. Designing a practical public-key construction that simultaneously supports efficient public verification, strong cryptographic security, and high-quality text generation remains an interesting direction for future research.

\subsection{Generation Examples}
\label{app:generation-examples}

Following the completion-oriented presentation used in prior work
\cite{fairoze2023publicly}, Table~\ref{tab:generation-examples} places each
prompt beside a plain completion and a watermarked completion produced by
\scheme{}. We reviewed all 30 completion pairs for relevance, coherence,
fluency, and repetition, and selected ten strong pairs spanning narrative,
scientific, social, and policy-oriented prompts. All examples use Llama-3-8B
with a 32-bit payload expanded into 96 random degree-3 fountain equations and
a maximum of 500 new tokens. We show only an initial, coherent excerpt from
each output. Leading prompt-continuation scaffolding and Markdown delimiters
are omitted; the terminal ellipses are editorial truncation markers and were
not emitted by the model.

\begin{table*}[p]
	\centering
	\fontsize{7}{7.4}\selectfont
	\setlength{\abovecaptionskip}{0cm}
	\renewcommand{\arraystretch}{1.02}
	\setlength{\tabcolsep}{4pt}
	\caption{Ten selected high-quality completions from 30 generation tasks in
		the main experiment. IDs match the saved task indices. Plain and
		watermarked outputs were generated from the same prompt. Excerpts are
		truncated at a natural boundary and lightly normalized for space and
		typography.}
	\vspace{6pt}
	\begin{tabular}{@{}c p{0.175\textwidth} p{0.365\textwidth} p{0.365\textwidth}@{}}
		\hline \hline \noalign{\smallskip}
		ID & Prompt & Plain generation & \scheme{} watermarked generation \\
		\midrule
		000 &
		Explain how rainfall patterns affect crop growth in temperate regions. &
		Rainfall patterns play a crucial role in determining crop growth and
		productivity in temperate regions. The timing, amount, and distribution of
		rainfall affect soil moisture, seed germination, nutrient uptake, and
		ultimately crop yield. Adequate soil moisture supports root development,
		whereas drought or excess rainfall can stress crops and reduce
		yields.\ldots{} &
		Rainfall patterns have a significant impact on crop growth in temperate
		regions. During the growing season, sufficient rainfall promotes healthy
		development, water uptake, and nutrient uptake, whereas inadequate
		rainfall can expose crops to drought stress, reduced growth, lower yields,
		or even crop failure.\ldots{} \\
		\midrule
		005 &
		Write a concise overview of how vaccines train the immune system. &
		Vaccines train the immune system by exposing it to a weakened or altered
		form of a pathogen. The immune system produces antibodies that bind to the
		pathogen and remembers it, enabling a faster and more effective response
		during a later exposure and providing long-term protection.\ldots{} &
		Vaccines stimulate the immune system to recognize and respond to specific
		pathogens. They induce antibody production, activate immune cells, and
		create immunological memory. Memory T and B cells then support a rapid
		response to subsequent exposures.\ldots{} \\
		\midrule
		008 &
		Write a short essay about how urban parks influence public health. &
		Urban parks have become an essential component of city life, offering a
		tranquil escape from the hustle and bustle of urban living. They can
		increase physical activity, improve mental health, and reduce the risk of
		chronic diseases such as diabetes and heart disease. Walking and biking
		trails, playgrounds, and open spaces make exercise more accessible to city
		residents.\ldots{} &
		Urban parks are an often-overlooked but vital component of a city's
		infrastructure, providing recreational space and a multitude of health
		benefits. Safe and accessible places to exercise help combat obesity and
		improve circulation, mood, and community well-being. Greenery also offers
		an oasis from urban stress and can reduce symptoms of anxiety and
		depression.\ldots{} \\
		\midrule
		010 &
		Explain why ocean currents matter for global climate. &
		Ocean currents play a crucial role in regulating global climate because
		they distribute heat around the planet. The oceans absorb atmospheric
		heat and circulate it to different regions, helping maintain global
		temperatures and shaping regional climate patterns. This redistribution
		also mitigates extreme temperature variations between regions.\ldots{} &
		Ocean currents transport heat from equatorial regions toward the poles,
		influencing global temperature and weather patterns. Deep thermohaline
		circulation and wind-driven surface currents redistribute heat, moisture,
		and nutrients across the planet. Surface currents such as the Gulf Stream
		therefore help regulate local and regional climates.\ldots{} \\
		\midrule
		013 &
		Explain how inflation affects household purchasing power. &
		Inflation reduces household purchasing power when prices rise faster than
		income: the same amount of money buys fewer goods and services.
		Households may shift spending toward necessities, reduce discretionary
		purchases, or draw down savings. Inflation can also erode the real value
		of savings and create uncertainty about future consumption.\ldots{} &
		Inflation is a sustained increase in the general price level of goods and
		services. As prices rise, the value of money decreases, so households can
		buy less with the same income and may need to adjust their consumption
		patterns and financial plans. Rising rent, food, and utility costs can
		require difficult budget tradeoffs.\ldots{} \\
		\midrule
		016 &
		Write an accessible explanation of quantum entanglement for high school
		students. &
		Quantum entanglement describes particles whose properties become linked
		so that the state of one cannot be described independently of the others,
		even when they are separated by large distances. It is one of the
		fundamental ideas of quantum mechanics, where properties are represented
		by probabilities rather than the definite values familiar from classical
		physics.\ldots{} &
		Quantum physics studies really small things like atoms and particles that
		do not follow the familiar rules of everyday objects. One of its unusual
		rules is called ``entanglement'': a special connection between two
		particles or systems. Measurements of the two can remain strongly
		correlated even when the particles are very far apart.\ldots{} \\
		\midrule
		019 &
		Write a short narrative about a teacher adapting to a classroom full of
		AI tools. &
		As I stepped into the classroom on my first day of teaching, I felt
		overwhelmed by the sea of new devices and screens. My students set up
		their workstations with ease, while I struggled to adapt to the school's
		new AI-powered learning ecosystem. After several frustrating lessons, I
		began experimenting with tutorials, virtual reality, and other creative
		tools.\ldots{} &
		Mrs.\ Patel's alarm blared at 6:00 AM, signaling the start of a new school
		year. As she poured coffee, her mind returned to the summer's whirlwind
		training sessions on incorporating AI tools. She welcomed innovation, but
		this year's emphasis made her anxious. She soon found that language
		analysis and adaptive-learning tools could complement, rather than
		replace, her teaching.\ldots{} \\
		\midrule
		023 &
		Write a policy memo about reducing energy waste in office buildings. &
		Subject: Reducing Energy Waste in Office Buildings. As part of our
		organization's commitment to sustainability, I propose a policy to reduce
		energy consumption, lower operating costs, and minimize environmental
		impact through efficient practices and technologies. The measures include
		efficient lighting, appliances, elevators, and HVAC systems.\ldots{} &
		Title: Enhancing Energy Efficiency in Office Buildings: A Policy
		Memorandum. The policy encourages office buildings to reduce energy waste
		and targets a reduction of at least 15\% within two years. Education and
		training programs promote employee engagement and behavioral
		change, while annual audits and real-time energy monitoring identify
		further opportunities for improvement.\ldots{} \\
		\midrule
		026 &
		Write a short story about a city that loses access to the internet for
		one day. &
		Title: The Day the Internet Died. I had just finished a particularly
		uninspiring day at work and was heading home to nurse my creative wounds.
		As an artist, I relied heavily on the internet for inspiration and sales.
		Without it, I felt like a fish out of water. Back at my studio, I realized
		that my computer was offline too, threatening both an exhibition deadline
		and my connection with other artists.\ldots{} &
		The sun rose over the concrete jungle, casting a warm glow over the
		towering skyscrapers of Nova City. It was just another ordinary day for
		the millions of residents who called the metropolis home. Then their
		phones stopped connecting, and Nova City realized it had lost access to
		the world wide web. Emails, bank accounts, and social media became
		unreachable as technicians rushed to diagnose the outage.\ldots{} \\
		\midrule
		029 &
		Discuss the risks of overreliance on automated decision systems. &
		A key concern is that automated systems may not always provide accurate
		or unbiased information. In healthcare they may recommend ineffective
		treatments; in finance they may act on incomplete data; and in education
		they may reproduce unfair grading biases. Opaque decisions can also make
		errors difficult to understand, challenge, or correct.\ldots{} &
		Overreliance on automated decision systems poses significant risks to
		individuals and society, especially in criminal justice, employment, and
		healthcare. Historical training data can perpetuate bias, while opaque
		models make accountability difficult. Programming errors and unintended
		consequences can further harm people in these high-stakes settings.\ldots{} \\
		\noalign{\smallskip} \hline \hline
	\end{tabular}
	\label{tab:generation-examples}
\end{table*}

\section{Related Work}
\label{sec:related-work}

\paragraph{Statistical and distribution-aware LLM watermarks.}
The green-list construction of Kirchenbauer et al.~\cite{kirchenbauer2023watermark}
made token-level statistical watermarking a practical provenance mechanism for
autoregressive LMs. Follow-up work studies its reliability under rewriting,
span mixing, and low-entropy prompts~\cite{kirchenbauer2024reliability}, and
develops distortion-free, unbiased, or optimized samplers to improve the
quality--detectability tradeoff~\cite{kuditipudi2024robust,hu2024unbiased,
christ2024undetectable,li2024statistical,giboulot2024watermax,
zhao2023provable}. Other work uses lexical redundancy or probability-aware
token selection to limit utility loss~\cite{watme2024,ren2024riw}, while
entropy-weighted detection explicitly addresses low-entropy
positions~\cite{lu2024entropy}. \scheme{} likewise operates at sampling time,
but its goal is not a secret-key hypothesis test: its detector scores public
token evidence against signed registry records.

\paragraph{Robust and semantic watermarks.}
Token-local schemes can lose synchronization after editing or paraphrasing.
Semantic-invariant and semantics-based methods therefore derive watermark
decisions from a semantic representation of the preceding text rather than
only a short token hash~\cite{liu2024sir,ren2024semamark}. Context-aware
balanced lists similarly aim to preserve semantic quality while increasing
robustness~\cite{guo2024context}; SemStamp moves the watermark unit to
sentence-level semantic regions~\cite{hou2024semstamp}; and PostMark shows a
post-hoc, black-box-compatible route to robust
watermarking~\cite{chang2024postmark}. Subsequent analysis also demonstrates
that purported paraphrase robustness must be evaluated against
reverse-engineering-aware attacks~\cite{rastogi2024paraphrasing}. These works
primarily improve the resilience of a keyed watermark signal. In contrast, our
local transcript distributes evidence for a registry-authorized payload; the
formal claim is about public evidence aggregation, not semantic invariance.

\paragraph{Payload-bearing and public verification schemes.}
Multi-bit methods allocate positions or layers to carry user-level information
and recover it under corruption~\cite{yoo2024multibit,feng2025bimark}. The
publicly-detectable construction of Fairoze et al.~\cite{fairoze2023publicly}
embeds a publicly verifiable signature with rejection sampling and error
correction. UPV uses separate generation and detection networks to target
unforgeable public verification~\cite{liu2024upv}, whereas PVMark makes a
private-key detector publicly auditable through zero-knowledge proofs
\cite{duan2026pvmark}. \scheme{} differs by authorizing a short payload in a
signed registry and representing it as many one-bit equation votes. Therefore,
verification need not exactly decode a long cryptographic string; it tests
whether noisy text supports one of the bounded, signed hypotheses. The
public-checking abstraction is motivated by the broader verifiable-random-
function paradigm~\cite{micali1999vrf}, but a concrete instantiation remains
future work.

\paragraph{Forgery, spoofing, and cryptographic foundations.}
Watermark removal is not the only security objective: color-aware token
substitutions can exploit a recovered watermark partition~\cite{wu2024bypassing},
while Bileve distinguishes source tracing from fine-grained integrity in the
presence of spoofing~\cite{zhou2024bileve}. More broadly, complexity-theoretic
steganography formalizes undetectability relative to a sampled cover
channel~\cite{hopper2002steganography}, and public-key steganography shows how
hidden communication can be established without a pre-shared secret
\cite{vonahn2004publickey}. \scheme{} adopts the public-verification motivation
but targets provenance evidence rather than covert communication: its
registry is explicit, and its security boundary separates issuer authorization
from edit-tolerant statistical detection.

\paragraph{Evaluation and robustness assessment.}
Recent benchmarks emphasize that watermarking must be assessed jointly for
text quality, detection length, and resistance to transformations rather than
by an isolated detector score~\cite{piet2024markmywords,tu2024waterbench}.
WaterPark further systematizes watermarker and attack design choices, exposing
their effect on adversarial robustness~\cite{liang2025waterpark}. Following
this perspective, our evaluation reports ordinary and wrong-context negative
controls together with deletion, substitution, truncation, copy-paste, and
suffix attacks. We additionally expose equation coverage, rank, registry
agreement, matching mass, and fallback behavior, because these diagnostics
explain whether a public-verification decision is supported by independent
authorized evidence.

\paragraph{Fountain-style evidence aggregation.}
LT codes and related fountain constructions recover data from redundant,
randomized equations rather than a fixed packet order~\cite{luby2002lt}.
\scheme{} borrows the one-bit equation viewpoint, but does not use it as a
standalone erasure-code decoder. Each equation is a noisy statistical vote
derived from a token, and the registry restricts the payload hypotheses before
the verifier scores their agreement. This distinction is what allows the
construction to tolerate fallback sampling and local text edits without
requiring exact recovery of an embedded bitstream.

\section{Conclusion} \label{sec}

We introduced \scheme{}, a public-key watermarking framework that separates
payload authorization from noisy textual evidence. A signed registry record
authorizes a short context-bound payload, generation distributes evidence
across token-level equations, and public verification aggregates those votes
instead of exactly decoding a long cryptographic string. In our $32$-bit
configuration, all unmodified outputs were detected, edited outputs retained
valid rates of $0.967$--$1.000$, and all evaluated plain and wrong-context
controls were rejected. The parameter sweeps favor short payloads, moderate
generation lengths, and local context binding.

An acceptance decision establishes issuer-authorized evidence, not
full-document integrity or token-by-token authorship. A complete deployment
therefore still requires a concrete unforgeable public-checking primitive and,
where stronger provenance is needed, complementary signatures or span-level
localization. Within this boundary, \scheme{} shows that offline public
verification can be combined with edit-tolerant statistical evidence through
signed registry payloads.

\bibliographystyle{plain}
\bibliography{references}

\balance
	
\end{document}